\documentclass[a4paper,reqno]{amsart}
\usepackage[margin=2.5cm,papersize={19.5cm,24.8cm}]{geometry}
\usepackage[all]{xy}           
\usepackage{amssymb}           
\usepackage{hyperref}
\usepackage{eucal}
\usepackage[dvips]{graphics}
\usepackage{graphicx}
\usepackage{subfig}
\DeclareGraphicsExtensions{.jpg,.png,.pdf}
\usepackage{exscale}
\usepackage{color}

\numberwithin{equation}{section}

\newtheorem{definition}{Definition}[section]

\newtheorem{proposition}[definition]{Proposition}

\newtheorem{remarkth}[definition]{Remark}

\renewcommand{\emph}[1]{{\bfseries\itshape{#1}}}

\newcommand{\lvec}[1]{\overleftarrow{#1}}
\newcommand{\rvec}[1]{\overrightarrow{#1}}

\makeatletter
\newcommand\prol{\@ifstar{\@proldf}{\@prolpf}}  
\def\@prolpf{\@ifnextchar[{\@prolpf@wrt}{\@prolpf@}}
\def\@prolpf@wrt[#1]#2{\@ifnextchar[{\@prolpf@wrt@at{#1}{#2}}{\@prolpf@wrt@{#1}{#2}}}
\def\@prolpf@wrt@at#1#2[#3]{\prolsymbol^{#1}_{#3}#2}
\def\@prolpf@wrt@#1#2{\prolsymbol^{#1}#2}
\def\@prolpf@#1{\@ifnextchar[{\@prolpf@at{#1}}{\@prolpf@@{#1}}}
\def\@prolpf@at#1[#2]{\prolsymbol_{#2}#1}
\def\@prolpf@@#1{\prolsymbol#1}
\def\@proldf{\@ifnextchar[{\@proldf@wrt}{\@proldf@}}
\def\@proldf@wrt[#1]#2{\@ifnextchar[{\@proldf@wrt@at{#1}{#2}}{\@proldf@wrt@{#1}{#2}}}
\def\@proldf@wrt@at#1#2[#3]{\prolsymbol^{*#1}_{#3}#2}
\def\@proldf@wrt@#1#2{\prolsymbol^{*#1}#2}
\def\@proldf@#1{\@ifnextchar[{\@proldf@at{#1}}{\@proldf@@{#1}}}
\def\@proldf@at#1[#2]{\prolsymbol^*_{#2}#1}
\def\@proldf@@#1{\prolsymbol^*#1}
\def\prolsymbol{\mathcal{T}}
\makeatother

\begin{document}

\title[Poisson-Lie groups, bi-Hamiltonian systems and integrable deformations]{Poisson-Lie groups, bi-Hamiltonian systems and integrable deformations}

\author[A.\ Ballesteros]{Angel Ballesteros}
\address{A.\ Ballesteros:
Departamento de F{\'\i}sica, Universidad de Burgos \\
Burgos, Spain}
\email{angelb@ubu.es}

\author[J.\ C.\ Marrero]{Juan C.\ Marrero}
\address{Juan C.\ Marrero:
ULL-CSIC Geometr\'{\i}a Diferencial y Mec\'anica Geom\'etrica\\
Departamento de Matem\'aticas, Estad{\'\i}stica e IO, Secci\'on de
Ma\-te\-m\'a\-ti\-cas, Universidad de La Laguna \\
 La Laguna, Tenerife, Canary Islands, Spain} \email{jcmarrer@ull.edu.es}

\author[Z. Ravanpak]{Zohreh Ravanpak}
\address{Z.\ Ravanpak: 
Department of Mathematics, Faculty of Science \\
Azarbaijan Shahid Madani University \\ 
Tabriz, Iran} 
\email{z.ravanpak@azaruniv.edu}

\thanks{A.B. has been partially supported by Ministerio de Econom\'{i}a y Competitividad (MINECO, Spain) under grants MTM2013-43820-P and MTM2016-79639-P (AEI/FEDER, UE), and by Junta de Castilla y Le\'on (Spain) under grants BU278U14 and VA057U16. J.C.M. has been partially supported by Ministerio de Econom\'{i}a y Competitividad (MINECO, Spain) under grant MTM 2015-64166-C2-2P. Z.R. has been partially supported by Ministry of Science Research and Technology (MSRT, Iran) under grant 2015-215401}



\begin{abstract}
Given a Lie-Poisson completely integrable bi-Hamiltonian system on $\mathbb{R}^n$, we present a method which allows us to construct, under certain conditions, a completely integrable bi-Hamiltonian deformation of the initial Lie-Poisson system on a non-abelian Poisson-Lie group $G_\eta$ of dimension $n$, where $\eta \in \mathbb{R}$ is the deformation parameter. Moreover, we show that from the two multiplicative (Poisson-Lie) Hamiltonian structures on $G_\eta$ that underly the dynamics of the deformed system and by making use of the group law on $G_\eta$, one may obtain two completely integrable Hamiltonian systems on $G_\eta \times G_\eta$. By construction, both systems admit reduction, via the multiplication in $G_\eta$, to the deformed bi-Hamiltonian system in $G_\eta$. The previous approach is applied to two relevant Lie-Poisson completely integrable bi-Hamiltonian systems: the Lorenz and Euler top systems.      

\end{abstract}

\vspace{1cm}

\maketitle

\section{Introduction}

It is well-known that a Hamiltonian system on a symplectic manifold $M$ of dimension $2r$ is (Liouville) completely integrable if there exist $r$ first integrals that pairwise commute and which are functionally independent in a dense open subset $U$ of $M$. In such a case, $U$ admits a Lagrangian foliation and the solutions of the Hamiltonian dynamics live in the leaves of this foliation (see \cite{Ar}). The previous notion may be extended, in a natural way, for the more general case when the phase space $M$ is a Poisson manifold $P$, not necessarily symplectic (for more details, see \cite{LaMiVa}).
Following this approach, a dynamical system on a manifold $P$ is said to be bi-Hamiltonian if admits two Hamiltonian descriptions with respect to two compatible Poisson structures on $P$. Bi-hamiltonian and completely integrable Hamiltonian systems are closely related since, under certain conditions, a bi-Hamiltonian system is completely integrable (see, for instance, \cite{KoMa}).   

We also recall that for a multiplicative Poisson structure on a Lie group $G$,  the multiplication is a Poisson epimorphism. In these conditions, the dual space ${\frak g}^*$ of the Lie algebra ${\frak g}$ of $G$ admits a Lie algebra structure in such a way that the couple $({\frak g}, {\frak g}^*)$ is a Lie bialgebra. In fact, the Lie algebra structure on ${\frak g}^*$ is defined by the dual map of an adjoint $1$-cocycle on ${\frak g}$ with values in $\Lambda^2 {\frak g}$. Conversely, an adjoint $1$-cocycle on a Lie algebra ${\frak g}$ with values in $\Lambda^2 {\frak g}$ whose dual map satisfies the Jacobi identity, induces a unique multiplicative Poisson structure on a connected simply connected Lie group with Lie algebra ${\frak g}$ (see~\cite{Dr}). Lie groups which are endowed with multiplicative Poisson structures are called Poisson-Lie groups and Lie-Poisson structures on the dual space of a Lie algebra ${\frak g}$ are examples of abelian Poisson-Lie groups (for more details see, for instance, \cite{Va}). Poisson-Lie groups are instances of Poisson coalgebras for which the comultiplication map is given by the group law, and the quantization of the former are the so-called quantum groups, which are the underlying symmetries of many relevant quantum integrable models (see, for instance,~\cite{CP,DrQG}).

As it was shown in~\cite{BaRa}, Poisson coalgebras can be systematically used in order to construct completely integrable Hamiltonian systems with an arbitrary number of degrees of freedom. Moreover, under this approach, deformations of Poisson coalgebras provide integrable deformations of the previous systems, and all constants of the motion can be explicitly obtained. Since then, this approach has been extensively used in order to construct different types of finite-dimensional integrable systems (see~\cite{jpcs,BaBlMu,LieH,BRcluster} and references therein) and several closely related constructions relying on the modification of the underlying Poisson coalgebra symmmetry have been also proposed in~\cite{comodule,Grabowski,Musso}.

Nevertheless, the generalization of the Poisson coalgebra approach to bi-Hamiltonian systems was still lacking, and the aim of this paper is to fill this gap by presenting a systematic approach for the construction of integrable deformations of bi-Hamiltonian systems, which will be based on the theory of multiplicative Poisson structures on Lie groups. 

To achieve this goal, we will firstly need an appropriate geometric interpretation of the results recently presented in~\cite{BaBlMu}. In particular, from a Lie-Poisson completely integrable Hamiltonian system defined on the dual space ${\frak g}^*$ of a Lie algebra ${\frak g}$ and in the presence of an arbitrary uni-parametric family $\{\Psi_{\eta}\}_{\eta \in \mathbb{R}}$ of adjoint $1$-cocycles whose dual maps satisfy the Jacobi identity, we will show that a Hamiltonian deformation of the initial system can be constructed on a connected and simply connected Lie group with Lie algebra ${\frak g}_{\eta}^*$. When $\eta$ approaches to zero, one recovers the initial system on ${\frak g}^*$ and, under certain conditions, the deformed system is also completely integrable. Now, by using that the multiplication in $G_\eta$ is associative and a Poisson epimorphism, one may obtain new  Hamiltonian systems with more degrees of freedom on $N$ copies of $G_\eta$, by coupling of the integrable Hamiltonian deformation in $G_\eta$. We will show that, by construction, these systems admit reduction, via the multiplication, to the deformed system in $G_\eta$. Moreover, under certain conditions, they are also completely integrable.

Secondly, if we want to generalize the previous construction when the initial completely integrable Hamiltonian system is bi-Hamiltonian with respect to two compatible Lie-Poisson structures $\Pi_0$ and $\Pi_1$ on $\mathbb{R}^n$, this implies that we have to deal with two compatible Lie algebra structures $[\cdot, \cdot]_0$ and $[\cdot, \cdot]_1$ on $\mathbb{R}^n$ which induce the two compatible Poisson structures $\Pi_0$ and $\Pi_1$. As we will show in this paper, the following important considerations and findings arise as a consequence of the bi-Hamiltonian approach:
\begin{itemize}
\item
In order to stay within the category of bi-Hamiltonian systems, we will have to impose that the integrable deformation of the initial system has to be bi-Hamiltonian with respect to two compatible multiplicative Poisson structures $\Pi_{0\eta}$ and $\Pi_{1\eta}$ on the {\em same} Lie group $G_\eta$. So, we should be able to find a common uni-parametric family of adjoint $1$-cocycles $\{\Psi_{\eta}\}_{\eta \in \mathbb{R}}$ for the Lie algebra structures $[\cdot, \cdot]_0$ and $[\cdot, \cdot ]_1$.

\smallskip

\item
When comparing with the prescriptions given by the generic Poisson coalgebra method used in~\cite{BaBlMu,BaRa}, now we have more constraints in choosing the two Hamiltonian functions $H_{0\eta}$ and $H_{1\eta}$ for the bi-Hamiltonian system on $G_{\eta}$. The reason is that $H_{0\eta}$  and $H_{1\eta}$ must be commuting functions for the Poisson structures $\Pi_{0\eta}$ and $\Pi_{1\eta}$. This, in some cases, fixes uniquely the definition of the Hamiltonians $H_{0\eta}$ and $H_{1\eta}$.

\smallskip

\item
The dynamics of the (completely integrable) Hamiltonian systems $(\Pi_{0\eta}, H_{0\eta})$ and  $(\Pi_{1\eta}, H_{1\eta})$ on $G_\eta$ coincide. However, the coupling, on one hand, of the system $(H_{0\eta}, \Pi_{0\eta})$ on $G_{\eta} \times G_{\eta}$ and the coupling, on the other hand, of the system $(H_{1\eta}, \Pi_{1\eta})$ on $G_{\eta} \times G_{\eta}$ can be used to produce two completely integrable Hamiltonian systems on the Lie group $G_\eta \times G_\eta$ which do not have, in general, the same dynamics. In other words, the method does not provide, in general, a bi-Hamiltonian system on the phase space $G_\eta \times G_\eta$.

\smallskip

\item
Nevertheless, the two completely integrable Hamiltonian systems on $G_\eta \times G_\eta$ admit reduction, via the multiplication, to the deformed bi-Hamiltonian system on $G_\eta$. So, for this reason, we can say that the Hamiltonian systems on $G_\eta \times G_\eta$ are `quasi-bi-Hamiltonian'.  Moreover, in the same way as in the general method~\cite{BaRa} and by using the associativity of the multiplication in $G_\eta$, one may extend this construction in order to obtain  two multiplicative completely integrable Hamiltonian systems on $N$ copies of $G_\eta$, with $N \geq 2$. By construction, these systems admit reduction (via the multiplication) to the deformed bi-Hamiltonian system on $G_\eta$ and will preserve for any $N$ its `quasi-bi-Hamiltonian'\ nature.

\smallskip

\item The method here presented is fully constructive and it could be applied to any Lie-Poisson completely integrable bi-Hamiltonian system such that both Lie algebra structures $[\cdot, \cdot]_0$ and $[\cdot, \cdot ]_1$ have a common 1-cocycle. Indeed, if this cocycle is found to be multiparametric, then we would obtain a multiparametric integrable deformation of the initial bi-Hamiltonian system.

\end{itemize}

The paper is structured as follows. In Section \ref{Bi-Ha-sy-Po-Lie-gr}, we will review some definitions and basic results on Poisson structures, Poisson-Lie groups and Poisson bi-Hamiltonian systems. In Section 3, we will present the systematic method to obtain integrable deformations of Lie-Poisson bi-Hamiltonian systems. For the sake of clarity, we will exemplify the method to the particular case when our initial dynamical system is a specific Lie-Poisson completely integrable bi-Hamiltonian system on $\mathbb{R}^4$. Our motivation for considering this system lies in the fact that its restriction to a submanifold of codimension $1$ is just an integrable limit of the well-known Lorenz dynamical system (see \cite{BaBlMu,Lo}). In Section 4 we will face the problem of the construction of the two completely integrable quasi-bi-Hamiltonian systems on $\mathbb{R}^{2n} = \mathbb{R}^n \times \mathbb{R}^n$, which will admit a reduction to the bi-Hamiltonian systems on $\mathbb{R}^n$ that have been presented in the previous section. The deformed counterpart of this construction leading to coupled systems on $G_\eta \times G_\eta$ is presented in Section 5. In order to show the fully constructive nature of the approach here introduced, in Section 6 we apply it to another relevant Lie-Poisson completely integrable bi-Hamiltonian system: an Euler top on $\mathbb{R}^3$. Finally, a concluding section closes the paper.


\section{Bi-Hamiltonian systems and Poisson-Lie groups}\label{Bi-Ha-sy-Po-Lie-gr}
In this section, we will review some definitions an basic results on Poisson-Lie groups and bi-Hamiltonian systems on Poisson manifolds (for more details, see \cite{Va}).

\subsection{Poisson manifolds and Lie-Poisson structures}
We recall that a Poisson structure on a manifold $M$ is a bivector field $\Pi$ on $M$ such that the Schouten-Nijenhuis bracket
$[\Pi, \Pi] = 0$. This is equivalent to defining a Lie algebra
structure on $C^{\infty}(M)$, whose bracket $\{\cdot,\cdot\}$ is called the Poisson bracket. The relation between the two definitions is given by the formula:
\[
\{f,g\}=\Pi(df,dg),\quad \forall f,g \in C^{\infty}(M).
\]
A Poisson structure $\Pi$ on $M$ determines in a natural way a vector bundle morphism $\Pi^{\sharp}: T^*M \to TM$ from the cotangent bundle $T^*M$ of $M$ to the tangent bundle $TM$. We also denote by $\Pi^{\sharp}$ the corresponding morphism between the space $\Omega^{1}(M)$ of $1$-forms on $M$ and the space ${\frak X}(M)$ of vector fields. 
The rank of the Poisson structure at the point $x \in M$ is just the rank of the linear map  $\Pi^{\sharp}(x): T_x^*M \to T_xM$. Since this linear map is skew-symmetric, the rank is always an even number. It is clear that $\Pi^{\sharp}(x): T_x^*M \to T_xM$ is not, in general, an isomorphism. In fact, a real $C^{\infty}$-function on $M$ is said to be a Casimir if $\Pi^{\sharp}(df) = 0$ or, equivalently,
\[
\{ f, g\} = 0, \; \; \; \forall g \in C^{\infty}(M).
\]

If the rank of a Poisson structure $\Pi$ on a manifold $M$ is maximum and equal to the dimension of $M$ then the structure is symplectic. This means that the manifold $M$ has even dimension $2n$ and it admits a closed $2$-form $\Omega$ which is non-degenerate, i.e., the vector bundle morphism $\flat_{\Omega}: TM \to T^*M$, induced by $\Omega$, is an isomorphism. In fact, in such a case, we have that $\Pi^{\sharp}$ is just the inverse morphism of $\flat_{\Omega}$. 

Another interesting class of Poisson structures are the so-called Lie-Poisson structures on $\mathbb{R}^n$. A Poisson structure on $\mathbb{R}^n$ is said to be Lie-Poisson if the bracket of two linear functions is again a linear function. So, if $(x_1, \dots, x_n)$ are the standard coordinates on $\mathbb{R}^n$, it follows that
\begin{equation}\label{Poisson-structure-eqs}
\{x_i, x_j \} = c_{ij}^k x_k, \; \; \; \mbox{ for } i, j \in \{1, \dots, n\},
\end{equation}
where $c_{ij}^k \in \mathbb{R}$. 

In fact, there exists a one-to-one correspondence between Lie-Poisson structures on $\mathbb{R}^n$ and Lie algebra structures on the same space. In fact, the Lie algebra structure $[\cdot, \cdot]$ on $\mathbb{R}^n$ associated with the Lie-Poisson structure characterized by (\ref{Poisson-structure-eqs}) is given by
\[
[x^{i}e_i, x^{j}e_j] = x^{i} x^{j} c_{ij}^k e_k.
\]     


\subsection{Poisson bi-Hamiltonian systems}
Let $\Pi$ be a Poisson structure on a manifold $M$. Then, a smooth real function $H$ (the Hamiltonian function) induces a vector field $X_{H}=\Pi^{\sharp}(dH)$ (the Hamiltonian vector field).
Hamilton\textquoteright s equations of motion for $H$ are:
 $$\dot{x}=\Pi^{\sharp}(dH) \equiv \{x,H\}.$$
So, solutions of the Hamilton equations are just the integral curves of $X_H$. The pair $(\{\cdot, \cdot\}, H)$ is said to be a Poisson Hamiltonian system.

Let $(\{\cdot, \cdot\}_{0}, H_0)$ be a Hamiltonian system. This system is said to be bi-Hamiltonian if there exists another compatible Poisson structure $\{\cdot, \cdot\}_1$ and a Hamiltonian function $H_1$ such that the corresponding Hamiltonian vector fields $X_{H_{0}}$ and $X_{H_{1}}$ coincide. This means that  
$$\dot{x}=\{x, H_0\}_{0} = \{x, H_{1}\}_{1}.$$
We recall that two Poisson structures $\Pi$ and $\Pi'$ are said to be compatible if the sum $\Pi + \Pi'$  is also a Poisson structure or, equivalently, if the Schouten-Nijenhuis bracket $[\Pi, \Pi']$ is zero.

There exist different notions of completely integrable Poisson Hamiltonian systems. In this paper, we will adopt the following definition. Let $\Pi$ be a Poisson structure, with Poisson bracket $\{\cdot, \cdot\}$, on a manifold $M$ of dimension $n$ such that the rank of $\Pi$ is constant and equal to $2r \leq n$ in a dense open subset of $M$. A Poisson Hamiltonian system $(\{\cdot, \cdot\}, H)$ is said to be completely integrable if there exist $\varphi_1, \dots, \varphi_{r-1} \in C^{\infty}(M)$ satisfying the two following conditions:
\begin{enumerate}  
\item
The functions $H$ and $\varphi_1, \dots, \varphi_{r-1}$ are functionally independent in a dense open subset of $M$, that is,
\[
dH \wedge d\varphi_1 \wedge \dots \wedge d\varphi_{r-1} \neq 0,
\]
in a dense open subset of $M$.
\item
They are first integrals of the Hamiltonian system that pairwise commute, i.e.,
\[
\{\varphi_j, H\} = 0 \quad \mbox{ and } \quad \{\varphi_j, \varphi_k\} = 0, \; \; \; \mbox{ for } j, k \in \{1, \dots, {r-1}\}.
\]
\end{enumerate}

\subsection{Poisson-Lie groups}
A multiplicative Poisson structure $\Pi$ on a Lie group $G$ is a Poisson structure such that the multiplication $m: G \times G \to G$ is a Poisson epimorphism or, equivalently, 
$$\Pi(gh) = (T_gr_h)(\Pi(g)) + (T_hl_g)(\Pi(h)), \; \forall g, h \in G.$$
where $r_h: G \to G$ and $l_g: G \to G$ are the right and left translation by $h$ and $g$, respectively. In this case, $G$ is called a Poisson-Lie group.

Note that $\mathbb{R}^n$ endowed with a linear Poisson structure is an abelian Poisson-Lie group.
 
For a multiplicative Poisson structure $\Pi$ on a Lie group $G$, the linear map $\psi = d_{\mathfrak e}\Pi: \mathfrak{g}\longrightarrow \wedge^{2}\mathfrak{g}$ is a $1$-cocycle, i.e.
   $$[\xi, \psi{\eta}] - [\eta, \psi{\xi}] -\psi[\xi,\eta]=0,\quad \forall\xi,\eta\in \mathfrak{g} $$
 and the dual map $\psi^{*}: \Lambda^2{\mathfrak g}^* \longrightarrow \mathfrak g^{*}$ is a Lie bracket on $\mathfrak g^{*}$. In other words, the couple $(\mathfrak g, \mathfrak g^{*})$ is a Lie bialgebra. Note that if $\Pi$ is a multiplicative Poisson structure on $G$ then $\Pi({\frak e}) = 0$, ${\frak e}$ being the identity element in $G$.
 
On the other hand, if $((\mathfrak g, [\cdot, \cdot]), (\mathfrak g^{*}, [\cdot, \cdot]^{*}))$ is a Lie bialgebra and $G$ is a connected simply-connected Lie group with Lie algebra ${\frak g}$, then $G$ admits a multiplicative Poisson structure $\Pi$ and $[\cdot, \cdot]^* = d_{\mathfrak e}^*\Pi $.   


\section{Bihamiltonian deformations and integrability}

In this section, we will present a systematic method in order to obtain integrable deformations of Lie-Poisson bi-Hamiltonian systems and we will use the Lorenz system as a guiding example.

We recall that our initial data is a dynamical system $D$ on $\mathbb{R}^n$, which is bi-Hamiltonian with respect to two compatible linear Poisson structures. In fact, in the examples here presented, the system $D$ is completely integrable. The aim of this construction is two-fold:
\begin{itemize}

\item [$\bullet$] Firstly, to construct a bi-Hamiltonian deformation $D_\eta$ of the dynamical system on $\mathbb{R}^n$, whose bi-Hamiltonian structure will be provided by a pair of multiplicative (Poisson-Lie) structures on a non-abelian Lie group $G_\eta$. We will see that, under certain conditions and for every $\eta \in \mathbb{R}$,  this method is fully constructive and systematic, and guarantees that the limit $\eta\to 0$ of $D_\eta$ is just the initial dynamical system on $\mathbb{R}^n$. 

\medskip

\item [$\bullet$] Secondly, under the assumption that the bi-Hamiltonian system on $G_\eta$ is completely integrable, we will try to construct two completely integrable Hamiltonian systems on the product Lie group $G_{\eta} \times G_{\eta}$ whose projection, via the multiplication $\cdot_{\eta}: G_\eta \times G_\eta \to G_\eta$ in $G_\eta$, is just the bi-Hamiltonian and completely integrable system on $G_\eta$.
\end{itemize}

In the sequel we will make this construction explicit by splitting it into several steps.

\subsection{The system $D$ and its bi-Hamiltonian structure.}\label{D-bi-Ha-1}  
The Lorenz dynamical system (see \cite{Lo}) is given by the differential equations  
\begin{equation}\label{Lorenz-eqs}
\begin{array}{rcl}
\dot{x}&=&\sigma(y-x),\\
\dot{y}&=&\rho x-xz-y,\\
\dot{z}&=&-\beta z+xy,
\end{array}
\end{equation}
where $\sigma$ and $\rho$ are the Prandtl and Rayleigh numbers, respectively, and $\beta$ is the aspect ratio. In \cite{StEu,TaWe} (see also \cite{GuNu}), the authors consider the conservative limit of the previous equations, which is obtained through the following rescalling 
$$t\rightarrow\epsilon t,\quad x\rightarrow \frac{1}{\epsilon}x,\quad y\rightarrow \frac {1}{\sigma\epsilon^2}y,\quad z\rightarrow \frac{1}{\sigma\epsilon^2}z,\quad \epsilon=\frac{1}{\sqrt{\sigma\rho}}$$
In the limit $\epsilon\rightarrow 0$, the system~(\ref{Lorenz-eqs}) reduces to
\begin{equation}\label{Lorenz-rescall}
\begin{array}{rcl}
\dot{x}&=&y,\\
\dot{y}&=& x(1-z),\\
\dot{z}&=&xy.
\end{array}
\end{equation}
Furthermore, the transformation
$$x=x_1,\quad y=\frac{x_2}{2}, \quad z=\frac{x_{3}+2}{2} 
$$
carries~\eqref{Lorenz-rescall} into:
\[
\dot{x_1}=\frac{x_2}{2},\quad \dot{x_2}=-x_1x_3,\quad \dot{x_3}=x_1x_2.
\]
It is straightforward to check that the previous system is bi-Hamiltonian with respect to the Poisson structures $\{\cdot, \cdot \}_{0}$ and $\{\cdot, \cdot \}_{a}$ in $\mathbb{R}^3$ which are characterized by
\[
\begin{array}{rclcrclcrcl}
\{x_1,x_2\}_{0} &= &\displaystyle -\frac{x_3}{2},\quad & \{x_1,x_3\}_{0} &= & \displaystyle  \frac{x_2}{2},\quad & \{x_2,x_3\}_0& = &0,\\[5pt]
\{x_1,x_2\}_{a}& = & \displaystyle \frac{1}{4},\quad & \{x_1,x_3\}_{a} & = & 0,\quad & \{x_2,x_3\}_a &= & \displaystyle -\frac{x_1}{2}.
\end{array}
\]
The corresponding Hamiltonian functions are
\begin{equation}
H_{0}=x_3-x_{1}^2 \qquad  \mbox{ and } 
\qquad
H_{1}=x_{2}^2+x_{3}^2,
\label{hamLorentz}
\end{equation}
respectively (for more details, see \cite{GuNu}).

Note that the Poisson structure $\{\cdot, \cdot \}_a$ on $\mathbb{R}^3$ is affine and the corresponding linear Poisson structure $\{\cdot, \cdot \}_l$ on $\mathbb{R}^3$ is given by
\[
\{x_1,x_2\}_{l}=0,\quad \{x_1,x_3\}_{l}=0,\quad \{x_2,x_3\}_l=-\frac{x_1}{2}.
\]
So, we can consider the extension $\{\cdot, \cdot \}_1$ to $\mathbb{R}^4$ of $\{\cdot, \cdot \}_l$ 
\begin{equation}\label{1-Poisson-bracket}
\{x_1,x_2\}_{1}=\frac{x_4}{4},\quad \{x_1,x_3\}_{1}=0,\quad \{x_2,x_3\}_1=-\frac{x_1}{2},\quad \{\cdot,x_4\}_1=0,
\end{equation}
which is a Lie-Poisson structure. In this way, the affine subspace $A$ defined by the equation $x_4 = 1$ is a Poisson submanifold, and the induced Poisson structure on $A$ is just $\{\cdot, \cdot\}_a$. In the same manner, we also denote by $\{\cdot, \cdot\}_0$ the trivial extension to $\mathbb{R}^4$ of the Poisson structure $\{\cdot, \cdot\}_0$, namely
\begin{equation}\label{0-Poisson-bracket}
\{x_1,x_2\}_{0}=-\frac{x_3}{2},\quad \{x_1,x_3\}_{0}=\frac{x_2}{2},\quad \{x_2,x_3\}_0=0,\quad \{\cdot,x_4\}_0=0,
\end{equation}
and the Casimir functions for the Poisson bracket $\{\cdot, \cdot\}_0$ are
\[
{\mathcal C}_0 = x_2^{2} + x_{3}^{2}, \qquad {\mathcal C}_0' = x_4.
\] 

The linear Poisson structures $\{\cdot, \cdot\}_0$ and $\{\cdot, \cdot\}_1$ on $\mathbb{R}^4$ turn out to be compatible, in the sense that we can define a one-parametric family of Lie-Poisson structures (a Poisson pencil)
$$\quad\{.,.\}_{\alpha}=(1-\alpha)\{.,.\}_{0}+\alpha\{.,.\}_{1},\quad \mbox{ with } \alpha \in \mathbb{R}$$
whose explicit brackets are given by
\begin{equation}\label{Lie-Poi-alpha}
\begin{array}{rclcrcl}
\displaystyle \{x_1,x_2\}_{\alpha} &= &\displaystyle \frac{\alpha x_4}{4}-(1-\alpha) \displaystyle \frac{x_3}{2}, \quad & 
\{x_1,x_3\}_{\alpha}& = &(1-\alpha)\displaystyle \frac{x_2}{2}, \\[5pt]
\{x_2,x_3\}_{\alpha}& = &(-\alpha)\displaystyle \frac{x_1}{2},\quad & \{.,x_4\}_{\alpha} & = & 0.
\end{array}
\end{equation}
Obviously, if $\{X_1, X_2, X_3, X_4\}$ is the canonical basis of $\mathbb{R}^4$ then the corresponding Lie bracket $[\cdot, \cdot]_{\alpha}$ on $\mathbb{R}^4$ is given by
\begin{equation}
\begin{array}{rclcrcl}
[X_1,X_2]_{\alpha}& = & \displaystyle \frac{\alpha X_4}{4}-(1-\alpha)\displaystyle \frac{X_3}{2}, \quad &
[X_1,X_3]_{\alpha} & = &(1-\alpha)\displaystyle \frac{X_2}{2},\\[5pt]
[ X_2,X_3] _{\alpha}& = & (-\alpha)\displaystyle \frac{X_1}{2},\quad & [\cdot,X_4]_{\alpha}& = &0.
\end{array}
\label{pencil}
\end{equation}

Therefore, we can say that the dynamical system $D$
\begin{equation}\label{bi-Hamil-abelian}
\dot{x_1}=\frac{x_2x_4}{2},\qquad \dot{x_2}=-x_1x_3,\qquad \dot{x_3}=x_1x_2, \qquad \dot{x_4} = 0,
\end{equation}
is bi-Hamiltonian with respect to the Lie-Poisson structures $\{\cdot, \cdot\}_0$ and $\{\cdot, \cdot\}_1$ with Hamiltonian functions given by $H_0 = x_3 x_4 - x_1^2$ and $H_1 = x_2^2 + x_3^2$. From the previous considerations, we also deduce that this bi-Hamiltonian system is completely integrable.
Note that the original Lorenz system is recovered within the submanifold $x_4=1$.

In general, starting from a dynamical system $D$ on $\mathbb{R}^n$, the first task consists in finding two compatible linear Poisson structures $\{\cdot, \cdot \}_0$ and $\{\cdot, \cdot \}_1$ such that our dynamical system is bi-Hamiltonian with respect to these two Poisson structures. This means that there exist two Hamiltonian functions $H_0: \mathbb{R}^n \to \mathbb{R}$ and $H_1: \mathbb{R}^n \to \mathbb{R}$ and the evolution of an observable $\varphi \in C^{\infty}(\mathbb{R}^n)$ is given by
\[
\dot{\varphi} = \{\varphi, H_0\}_0 = \{\varphi, H_1\}_{1}.
\]
In fact, in the Lorenz system that we have just described we observe that:
\begin{itemize}
\item
The Hamiltonian $H_{0}$ (respectively, $H_1$) is a Casimir function ${\mathcal C}_1$ (respectively, ${\mathcal C}_0$) for $\{\cdot, \cdot\}_1$ (respectively, $\{\cdot, \cdot\}_0$).
\item
The rank of the Poisson structures $\{\cdot, \cdot \}_0$ and $\{\cdot, \cdot\}_1$ satisfies the following condition
\[
rank \{\cdot, \cdot\}_0 = rank \{\cdot, \cdot \}_1= 2r
\]
in a dense open subset of $\mathbb{R}^n$. In the previous example $r = 1$ and $n = 4$.
\item
The Hamiltonian systems $(\{\cdot, \cdot\}_0, H_0)$ and $(\{\cdot, \cdot\}_1, H_1)$ are completely integrable. 
\end{itemize}
In general, we will denote by $\{{\mathcal C}_i, {\mathcal C}_i^j \}_{j = 1, \dots, n-2r-1}$ the Casimir functions for $\{\cdot, \cdot \}_i$, with $i=0, 1$, and by $\{\varphi^j\}_{j = 1, \cdots, r-1}$ the set of common first integrals for the Hamiltonian systems $(\{\cdot, \cdot\}_0, H_0)$ and $(\{\cdot, \cdot\}_1, H_1)$.


\subsection{Construction of the bi-Hamiltonian system $D_\eta$.}\label{bi-Ha-D-eta}  
Let $[\cdot, \cdot ]_{0}$ (respectively, $[\cdot, \cdot ]_{1}$) be the Lie algebra structure on $\mathbb{R}^n$ associated with the linear Poisson bracket $\{\cdot, \cdot\}_0$ (respectively, $\{\cdot, \cdot\}_1$). Then, we have to find a non-trivial common adjoint $1$-cocycle $\psi_{\eta}:\mathbb R^n \to \wedge^2 \mathbb R^n$, with $\eta \in \mathbb{R}$, for the Lie algebras $(\mathbb R^n,[\cdot,\cdot]_0)$ and $(\mathbb R^n,[\cdot,\cdot]_1)$ and with the initial condition $\psi_0 = 0$. 

In doing so, we deduce the following result. 
\begin{proposition}
Let $\psi_{\eta}: \mathbb{R}^n \to \Lambda^2(\mathbb{R}^n)$ be a common adjoint $1$-cocycle for the compatible Lie algebras $(\mathbb{R}^n, [\cdot, \cdot]_0)$ and $(\mathbb{R}^n, [\cdot, \cdot]_1)$ and $G_{\eta}$ a connected simply-connected Lie group with Lie algebra $(\mathbb{R}^n, [\cdot, \cdot]^*_{\eta} = \psi_\eta^*)$. If $\{\cdot, \cdot\}_{0\eta}$ and $\{\cdot, \cdot\}_{1\eta}$ are the multiplicative Poisson brackets on $G_\eta$ associated with the $1$-cocycle $\psi_\eta: \mathbb{R}^n \to \Lambda^2(\mathbb{R}^n)$ then $\{\cdot, \cdot\}_{0\eta}$ and $\{\cdot, \cdot\}_{1\eta}$ are compatible.
\end{proposition}
\begin{proof}
It is a consequence of the following general result. If $H$ is a connected Lie group with Lie algebra $\frak{h}$ and $\{\cdot, \cdot\}_{0}$, $\{\cdot, \cdot\}_{1}$ are two multiplicative Poisson brackets on $H$ then the Poisson brackets are compatible if and only if the dual Lie algebras $(\frak{h}^*, [\cdot, \cdot]_0^*)$ and $(\frak{h}^*, [\cdot, \cdot]_1^*)$ are compatible.
\end{proof}

Next, we have to find a Casimir function $\mathcal C_{0\eta}$ (resp., $\mathcal C_{1\eta}$) for the Poisson bracket $\{\cdot,\cdot \}_{0\eta}$ (resp., $\{\cdot,\cdot \}_{1\eta}$) on $G_{\eta}$ in such a way that:
\begin{itemize}
\item [(i)]The Hamiltonian systems $(\{\cdot,\cdot \}_{0\eta}, H_{0\eta}:=\mathcal C_{1\eta})$ and  $(\{\cdot,\cdot \}_{1\eta}, H_{1\eta}:=\mathcal C_{0\eta})$ coincide, that is, we have a bi-Hamiltonian system $D_\eta$ on the Lie group $G_{\eta}$.
\item [(ii)]
This system is a $\eta$-deformation of the original bi-Hamiltonian system $D$, {\it i.e.},
$$\lim_{\eta \rightarrow 0}\{\cdot,\cdot \}
_{0\eta}=\{\cdot,\cdot\}_0, \; \; \; \lim_{\eta \rightarrow 0}\{\cdot,\cdot \}
_{1\eta}=\{\cdot,\cdot\}_1$$
and
$$\lim_{\eta \rightarrow 0}H
_{0\eta}=H_0, \; \; \; \lim_{\eta \rightarrow 0}H
_{1\eta}=H_1.$$
\end{itemize}
Moreover, in our examples, the Hamiltonian systems $(\{\cdot,\cdot \}_{0\eta}, H_{0\eta})$ and $(\{\cdot,\cdot \}_{1\eta}, H_{1\eta})$ are completely integrable and the ranks of $\{\cdot,\cdot \}_{0\eta}$ and $\{\cdot,\cdot \}_{1\eta}$ are again $2r$ within a dense open subset of $G_{\eta}$. We will denote by $\{{\mathcal C}_{i\eta}, {\mathcal C}_{i\eta}^j\}_{j = 1, \dots, n-2r-1}$ the Casimir functions for the multiplicative Poisson bracket $\{\cdot, \cdot\}_{i\eta}$, with $i = 0, 1$, and by $\{\varphi_{\eta}^j\}_{j=1, \dots, r-1}$ the common first integrals for the Hamiltonian systems $(\{\cdot, \cdot\}_{i\eta}, H_{i\eta})$, with $i = 0, 1$. In fact, we have that
\[
\lim_{\eta \rightarrow 0}{\mathcal C}_{i\eta}^j = {\mathcal C}_i^j, \; \; \; \mbox{ with } j = 1, \dots, n-2r-1,
\]
and
\[
\lim_{\eta \rightarrow 0}\varphi_\eta^j = \varphi^j, \; \; \; \mbox{ with } j = 1, \dots, r-1.
\]


\subsubsection{The Lorenz 1-cocycle and its associated non-abelian group $G_\eta$.}
 A straightforward computation shows that a non-trivial admissible cocycle for the family of Lie algebras $\mathfrak{g}_{\alpha}$ (see \eqref{pencil}) is:
$$
\psi_{\eta}(X_1)=0,\quad \psi_{\eta}(X_2)=-\eta X_3\wedge X_4,\quad \psi_{\eta}(X_3)=\eta X_2\wedge X_4,\quad \psi_{\eta}(X_4)=0.
$$
 So, we have the family of Lie bialgebras $(\mathfrak{g}_\alpha,\psi_{\eta})$. The Lie bracket $[\cdot, \cdot]_{\eta}^*$ on $(\mathbb{R}^4)^* \simeq \mathbb{R}^4$ obtained from the dual cocommutator map is:
\[
[X^2,X^4]_{\eta}^*=\eta X^3,\quad [X^3,X^4]_{\eta}^*=-\eta X^2,
\]
the rest of basic Lie brackets being zero.

Now, let $G_{\eta}$ be the connected and simply-connected Lie group with Lie algebra $(\mathbb{R}^4, [\cdot, \cdot]_{\eta}^*)$. Then, one may prove that $G_{\eta}$ is diffeomorphic to $\mathbb{R}^4$ and the multiplication $\cdot_\eta$ of two group elements $g = (x_1, x_2, x_3, x_4)$ and $g'=(x_1', x_2', x_3', x_4')$ reads
$$
g\cdot_\eta g' =
( x_1+x'_1,x_2+x'_2\cos(\eta x_4)+x'_3 \sin(\eta x_4),x_3-x'_2\sin(\eta x_4)+x'_3\cos(\eta x_4),x_4+x'_4).
$$
Note that $G_0$ is the abelian Lie group and $G_\eta$, with $\eta\neq 0$, is isomorphic to the special euclidean group $SE(2)$.

A basis $\{\lvec X^1,\lvec X^2,\lvec X^3,\lvec X^4\}$ (resp., $\{\rvec X^1,\rvec X^2,\rvec X^3,\rvec X^4\}$) of left-invariant (resp., right-invariant) vector fields for $G_\eta$ is found to be
\[
\{\frac{\partial}{\partial x_1}, \cos(\eta x_4)\frac{\partial}{\partial x_2}-\sin(\eta x_4)\frac{\partial}{\partial x_3}, \sin(\eta x_4)\frac{\partial}{\partial x_2}+\cos(\eta x_4)\frac{\partial}{\partial x_3}, \frac{\partial}{\partial x_4}\}
\]
(resp., $\displaystyle \{\frac{\partial}{\partial x_1}, \frac{\partial}{\partial x_2}, \frac{\partial}{\partial x_3}, \eta x_3\frac{\partial}{\partial x_2}-\eta x_2\frac{\partial}{\partial x_3}+\frac{\partial}{\partial x_4}\}$).

Finally, the adjoint action $Ad:G_{\eta}\times{\mathfrak {g}}_{\eta} \longrightarrow {\mathfrak {g}}_{\eta}$ for the Lie group $G_{\eta}$ can be straightforwardly computed:
\[
\begin{array}{rcl}
Ad_{g}(X^1)&=& X^1,\\
Ad_{g}(X^2)&=& \cos(\eta x_4) X^2-\sin(\eta x_4)X^3, \\
Ad_{g}(X^3)&=& \sin(\eta x_4) X^2+\cos(\eta x_4)X^3,\\
Ad_{g}(X^4)&=& -\eta x_3 X^2+\eta x_2 X^3+X^4.
\end{array}
\]


\subsubsection{A Poisson-Lie group structure on $G_\eta$.}\label{Po-Li-gr-st}
By construction, the following family of non-trivial admissible 1-cocycles for the Lie algebra $(\mathbb{R}^4, [\cdot, \cdot]_{\eta}^*)$ is obtained as the dual of the commutator map~\eqref{pencil}, namely
\[
\begin{array}{rclcrcl}
\psi_{\alpha}(X^1) & = &\displaystyle -\frac{\alpha}{2}X^2 \wedge X^3,\quad &  \psi_{\alpha}(X^2)& = & \displaystyle \frac{(1-\alpha)}{2} X^1\wedge X^3, \\[7pt]
 \psi_{\alpha}(X^3) &= & \displaystyle -\frac{(1-\alpha)}{2} X^1\wedge X^2,\quad & \psi_{\alpha}(X^4) &= &\displaystyle \frac{\alpha}{4}X^1\wedge X^2.
\end{array}
\]
Denote by $\Pi_{\alpha \eta}$ the (unique) multiplicative Poisson structure on $G_{\eta}$ which is induced by the $1$-cocycle $\psi_{\alpha}$. In order to obtain $\Pi_{\alpha \eta}$, we consider the $1$-form $\gamma_{\psi_{\alpha}}$ on $G_\eta$ with values in $\Lambda^2{\mathfrak g}_{\eta}$ which is characterized  by the following relation
\[
\gamma_{\psi_{\alpha}}(\lvec{X})(g) = Ad_g(\psi_{\alpha}(X)), \; \; \mbox{ for } X\in {\mathfrak g}_{\eta} \mbox{ and } g \in G_\eta.
\]
As we know (see the proof of Theorem $10.9$ in \cite{Va}), $\gamma_{\psi_{\alpha}}$ is an exact $1$-form. So, there exists a unique function $\pi: G_\eta \to \Lambda^2{\mathfrak g}_{\eta}$ satisfying
\[
\pi({\mathfrak e}) = 0 \quad \mbox{ and } \quad d\pi = \gamma_{\psi_{\alpha}},
\]
where ${\mathfrak e} = (0, 0, 0, 0)$ is the identity element in $G_\eta$.

If we suppose that
\[
\pi(g) = \pi_{ij}(g) X^{i} \wedge X^{j}, \; \; \; \mbox{ for } g \in G_\eta,
\]
then the multiplicative Poisson structure $\Pi_{\alpha \eta}$ is given by
\[
\Pi_{\alpha \eta}(g) = \pi_{ij}(g) \rvec{X^{i}} \wedge \rvec{X^{j}}
\]
(see the proof of Theorem $10.9$ in \cite{Va}).

Applying the previous process, we deduce that the corresponding compatible
multiplicative Poisson brackets $\{\cdot, \cdot \}_{\alpha \eta}$ on $G_{\eta}$ are given by 
\begin{equation}
\begin{array}{rclcrcl}
\{x_1,x_2\}_{\alpha \eta}& = & \displaystyle \frac{\alpha}{4}\frac{\sin(\eta x_4)}{\eta}-(1-\alpha)\frac{x_3}{2}, & \{x_2,x_3\}_{\alpha \eta}& = &(-\alpha)\displaystyle \frac{x_1}{2},\\[5pt]
 \{x_1,x_3\}_{\alpha \eta}& = & (1-\alpha)\displaystyle \frac{x_2}{2}+\frac{\alpha}{4}\frac{\cos(\eta x_4)-1}{\eta},
 &  \{.,x_4\}_{\alpha \eta} & = & 0.
\end{array}
\label{defpencil}
\end{equation}
As we expected, $\lim _{\eta\rightarrow 0}\{.,.\}_{\alpha \eta}=\{.,.\}_{\alpha}$. Thus, we have constructed an $\eta$-deformation~\eqref{defpencil} of the Lie-Poisson bracket (\ref{Lie-Poi-alpha}). We stress that~\eqref{defpencil} is just a multiplicative Poisson-Lie structure on the noncommutative group $G_\eta$, while~(\ref{Lie-Poi-alpha}) can be thought of as a multiplicative structure on the abelian Lie group $\mathbb{R}^4$.


\subsubsection{Casimir functions and deformed bi-Hamiltonian structure.}
In the particular cases when $\alpha=0$ and $\alpha=1$, the $\eta$-deformations $\{\cdot, \cdot\}_{0\eta}$ and $\{\cdot, \cdot\}_{1\eta}$ of the Lie-Poisson brackets $\{\cdot, \cdot\}_0$ and $\{\cdot, \cdot\}_1$ have the form
\begin{equation}\label{0-eta-Poisson-bracket}
\{x_1,x_2\}_{0\eta}  = \displaystyle -\frac{x_3}{2}, 
\qquad
\{x_1,x_3\}_{0\eta} = \displaystyle \frac{x_2}{2},  
\qquad
\{x_2,x_3\}_{0\eta} = 0,
\end{equation}
and
\begin{equation}\label{1-eta-Poisson-bracket}
\{x_1,x_2\}_{1\eta}  =  \displaystyle \frac{\sin(\eta x_4)}{4\eta}, 
\qquad \{x_1,x_3\}_{1\eta} =  \displaystyle \frac{\cos(\eta x_4)-1}{4\eta}, 
\qquad
\{x_2,x_3\}_{1\eta}  = \displaystyle -\frac{x_1}{2}.
\end{equation}
The function $x_4$ is a Casimir for both Poisson structures. 
 Other Casimir functions for these two multiplicative Poisson structures are found to be
  \[
  \mathcal C_{0\eta}=x_2^2+x_3^2, \qquad
\mathcal C_{1\eta}=\frac{\sin(\eta x_4)}{\eta}x_3-\frac{\cos(\eta x_4)-1}{\eta}x_2-x_1^2,
 \]
where it becomes clear that $\lim_{\eta \rightarrow 0}\mathcal C_{0\eta}=H_1$ and $\lim_{\eta \rightarrow 0}\mathcal C_{1\eta}=H_0$. In fact, if we denote the Casimir functions $\mathcal C_{0\eta}$ and $\mathcal C_{1\eta}$ by $H_{1\eta}$ and $H_{0\eta}$, respectively, then the dynamical systems associated with the Hamiltonian systems $(G_{\eta}, \{\cdot, \cdot \}_{0\eta}, H_{0\eta})$ and  $(G_{\eta}, \{\cdot, \cdot \}_{1\eta}, H_{1\eta})$ coincide and define the deformed dynamical system $D_\eta$ given by
\begin{equation}\label{bi-Hamil-non-abelian}
\dot{x_1}=\frac{1}{2}\frac{\sin(\eta x_4)}{\eta}x_2+\frac{1}{2} \frac{\cos(\eta x_4)-1}{\eta}x_3,\quad 
\dot{x_2}=-x_1x_3,\quad \dot{x_3}=x_1x_2,\quad \dot{x_4}=0.
\end{equation}
In other words, this dynamical system $D_\eta$ on the Lie group $G_{\eta}$ is bi-Hamiltonian with respect to the compatible multiplicative Poisson structures $\{\cdot, \cdot\}_{0\eta}$ and $\{\cdot, \cdot\}_{1\eta}$. Indeed, this deformed bi-Hamiltonian system is completely integrable and the $\eta\to 0$ limit is just the $D$ system (\ref{bi-Hamil-abelian}). The preservation of the closed nature of the trajectories under
deformation is clearly appreciated in Figure 1, where the trajectories have been found by numerical integration.

\begin{figure}[h]
\subfloat[a][]
\centering
 \scalebox{0.6}[0.6]{\includegraphics{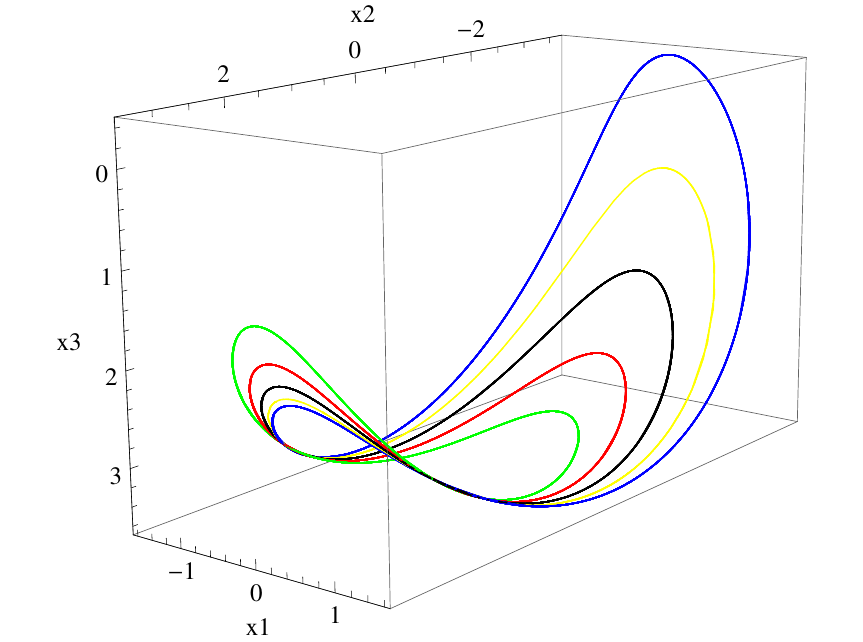}} 
\subfloat[b][]
{\includegraphics[scale=0.6]{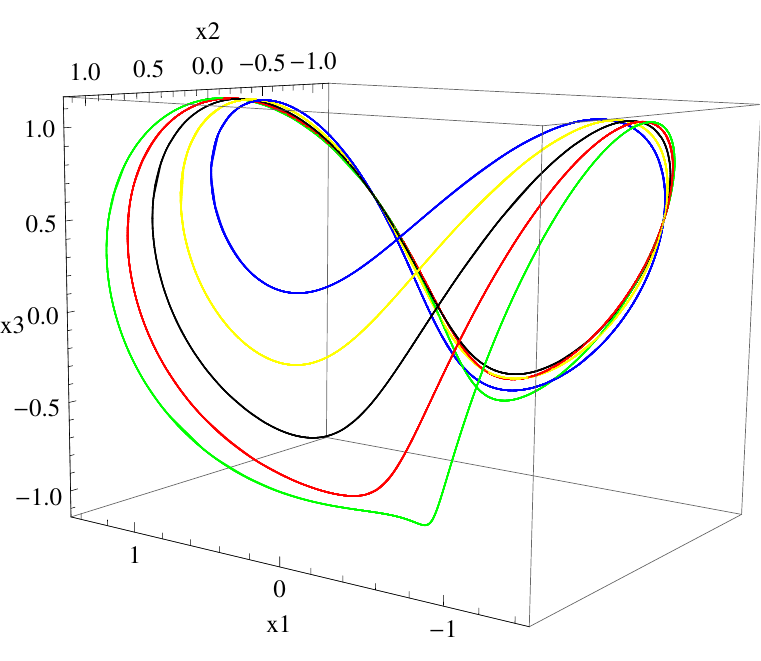}}
\caption{(A) Closed trajectories of the Lorenz system (\ref{bi-Hamil-abelian}) for $x_4=1$ and the initial data $x_1(0)=1,x_2(0)=2,x_3(0)=3$ (black line) and of the deformed Lorenz system (\ref{bi-Hamil-non-abelian}) with the same initial data, $x_4=1$, and $\eta=\frac{\pi}{4}$ (green), $\eta=-\frac{\pi}{4}$ (blue), $\eta=\frac{\pi}{8}$ (red), $\eta=-\frac{\pi}{8}$ (yellow). (B) The same figure as (A) but with initial data $x_1(0)=1,x_2(0)=-1,x_3(0)=0.5$. }
\label{fig:Lorenz3}
\end{figure}


\section{Coupled integrable deformations and non-abelian reduction}

In this last step, we present how to construct two coupled completely integrable Hamiltonian systems on the product Lie group $G_\eta \times G_\eta$ which admit reduction, via the multiplication $\cdot_{\eta} : G_\eta \times G_\eta \to G_\eta$, to the same bi-Hamiltonian system on $G_\eta$ that we denoted $D_\eta$ (see Section \ref{bi-Ha-D-eta}).
 
For this purpose, we will consider the product $\{\cdot, \cdot\}_{i\eta} \oplus \{\cdot, \cdot\}_{i\eta}$ of the multiplicative Poisson structure $\{\cdot, \cdot\}_{i\eta}$ on $G_\eta$ with itself, $i = 0, 1$. Then, we will obtain two multiplicative Poisson structures on $G_\eta \times G_\eta$ which, if there is not risk of confusion, we also denote by $\{\cdot, \cdot\}_{0\eta}$ and by $\{\cdot, \cdot\}_{1\eta}$, respectively. Note that
\[
rank \{\cdot, \cdot\}_{0\eta} = rank \{\cdot, \cdot\}_{1\eta} = 4r
\]
in a dense open subset of $G_\eta \times G_\eta$. Moreover,
\[
\{{\mathcal C}_{i\eta} \circ pr_1, {\mathcal C}_{i\eta}^j \circ pr_1, {\mathcal C}_{i\eta} \circ pr_2, {\mathcal C}_{i\eta}^j \circ pr_2\}_{j = 1, \dots, n-2r-1}
\]
are Casimir functions for $\{\cdot, \cdot\}_{i\eta}$, with $i = 0, 1$. Here, $pr_1: G_\eta \times G_\eta \to G_\eta$ and $pr_2: G_\eta \times G_\eta \to G_\eta$ are the canonical projections.

In addition, we also consider the coproduct of the Hamiltonian functions $H_{0\eta}$ and $H_{1\eta}$, that is, the Hamiltonian functions on $G_\eta \times G_\eta$ defined by $H_0 \circ \cdot_{\eta}$ and $H_1 \circ \cdot_{\eta}$, where $\cdot_{\eta} : G_\eta \times G_\eta \to G_\eta$ is the multiplication in $G_\eta$. If there is not risk of confusion, we will use the same notation $H_{0\eta}$ and $H_{1\eta}$ for the previous functions. 
 
We remark that $H_{0\eta}$ (respectively, $H_{1\eta}$) is a Casimir function for the Poisson bracket $\{\cdot, \cdot\}_{1\eta}$ (respectively, $\{\cdot, \cdot\}_{0\eta}$) on $G_\eta \times G_\eta$. Furthermore, we have that the multiplication $\cdot_{\eta}: G_\eta \times G_\eta \to G_\eta$ is a Poisson epimorphism between the Poisson manifolds $(G_\eta \times G_\eta, 
\{\cdot, \cdot\}_{0\eta} \oplus \{\cdot, \cdot\}_{0\eta})$ (respectively, $(G_\eta \times G_\eta, \{\cdot, \cdot\}_{1\eta} \oplus 
\{\cdot, \cdot\}_{1\eta})$) and $(G_\eta, \{\cdot, \cdot\}_{0\eta})$ (respectively, $(G_\eta, \{\cdot, \cdot\}_{1\eta})$). This implies the following result.
\begin{proposition}
The Hamiltonian systems $(\{\cdot, \cdot\}_{0\eta} \oplus \{\cdot, \cdot\}_{0\eta}, H_{0\eta} \circ \cdot_{\eta})$ and $(\{\cdot, \cdot\}_{1\eta} \oplus \{\cdot, \cdot\}_{1\eta}, H_{1\eta} \linebreak \circ \cdot_{\eta})$ on $G_\eta \times G_\eta$ admit reduction, via the multiplication $\cdot_{\eta}: G_\eta \times G_\eta \to G_\eta$, to the bi-Hamiltonian system $D_\eta$ on $G_\eta$.
\end{proposition}
From the previous result, the dynamical systems on $G_\eta \times G_\eta$ are said to be quasi bi-Hamiltonian systems.

We also remark that
\[
\lim_{\eta \rightarrow 0}(\{\cdot,\cdot \}
_{0\eta} \oplus \{\cdot,\cdot \}
_{0\eta})=\{\cdot,\cdot\}_0 \oplus \{\cdot,\cdot\}_0, \; \; \; \lim_{\eta \rightarrow 0}(\{\cdot,\cdot \}_{1\eta} \oplus \{\cdot,\cdot \}_{1\eta}=\{\cdot,\cdot\}_1 \oplus \{\cdot,\cdot\}_1
\]
and
\[
\lim_{\eta \rightarrow 0}(H
_{0\eta} \circ \cdot_{\eta})=H_0 \circ +, \; \; \; \lim_{\eta \rightarrow 0}(H
_{1\eta} \circ \cdot_{\eta}) =H_1 \circ +.
\]
Therefore, the Hamiltonian systems $(\{\cdot, \cdot\}_{i\eta} \oplus \{\cdot, \cdot\}_{i\eta}, H_{i\eta} \circ \cdot_{\eta})$, $i = \{0, 1\}$, may be considered as $\eta$-deformations of the quasi bi-Hamiltonian systems $(\{\cdot, \cdot\}_{i} \oplus \{\cdot, \cdot\}_{i}, H_{i} \circ +)$, $i = \{0, 1\}$, on $\mathbb{R}^n \times \mathbb{R}^n$. Note that these last systems admit reduction, via the sum $+: \mathbb{R}^n \times \mathbb{R}^n \to \mathbb{R}^n$, to the initial bi-Hamiltonian system on $\mathbb{R}^n$.

Moreover, we have that
\[
\{H_{1\eta} \circ \cdot_{\eta}, H_{0\eta} \circ \cdot_{\eta}, \varphi^k \circ \cdot_{\eta} \}_{ k= 1, \dots, r-1}
\]
are functionally independent first integrals of the Hamiltonian system $(\{\cdot, \cdot\}_{0\eta} \oplus \{\cdot, \cdot\}_{0\eta}, H_{0\eta} \circ \cdot_{\eta})$ that pairwise commute, and the same holds for 
\[
\{H_{0\eta} \circ \cdot_{\eta}, H_{1\eta} \circ \cdot_{\eta},  \varphi^k \circ \cdot_{\eta} \}_{k= 1, \dots, r-1}
\]
and the Hamiltonian system $(\{\cdot, \cdot\}_{1\eta} \oplus \{\cdot, \cdot\}_{1\eta}, H_{1\eta} \circ \cdot_{\eta})$. So, we conclude that
\begin{proposition}
If $r = 1$ the Hamiltonian systems $(\{\cdot, \cdot\}_{0\eta} \oplus \{\cdot, \cdot\}_{0\eta}, H_{0\eta} \circ \cdot_{\eta})$ and $(\{\cdot, \cdot\}_{1\eta} \oplus \{\cdot, \cdot\}_{1\eta}, H_{1\eta} \circ \cdot_{\eta})$ in $G_{\eta} \times G_{\eta}$ are completely integrable.
\end{proposition}
\begin{proof}
It follows using that $\{H_{0\eta} \circ \cdot_{\eta}, H_{1\eta} \circ \cdot_{\eta}\}$ are functionally independent first integrals for both Hamiltonian systems and, in addition, they pairwise commute.
\end{proof}
We remark that in the two examples presented in this paper, $r=1$.


\subsection{Deformed coupled Lorenz systems}

Denote by $\Pi_{0\eta}$ and $\Pi_{1\eta}$ the multiplicative Poisson structures on the Lie group $G_\eta$ associated with the Poisson brackets $\{\cdot, \cdot\}_{0\eta}$ and $\{\cdot, \cdot\}_{1\eta}$ given by (\ref{0-eta-Poisson-bracket}) and (\ref{1-eta-Poisson-bracket}), respectively, and by $(y, z) = ((y_1, y_2, y_3, y_4), (z_1, z_2, z_3, z_4))$ the standard coordinates on $G_\eta \times G_\eta \simeq \mathbb{R}^4 \times \mathbb{R}^4$.

Then, we can consider two multiplicative Poisson structures on  $G_\eta \times G_\eta$:

\noindent $\bullet$
The product of $\Pi_{0\eta}$ with itself, that is, the Poisson bracket on $G_\eta \times G_\eta$, which we also denote by $\{\cdot, \cdot\}_{0\eta}$, defined by
\[
\{y_1,y_2\}_{0\eta} = -\frac{y_3}{2}, 
\qquad
\{y_1,y_3\}_{0\eta}=\frac{y_2}{2},
\qquad
\{z_1,z_2\}_{0\eta} = -\frac{z_3}{2},
\qquad
\{z_1,z_3\}_{0\eta} = \frac{z_2}{2},
\]
the rest of the Poisson brackets of the coordinate functions being zero and

\noindent $\bullet$
The product of $\Pi_{1\eta}$ with itself, that is, the Poisson bracket on $G_\eta \times G_\eta$, which we also denote by $\{\cdot, \cdot\}_{1\eta}$, given by the non-vanishing Poisson brackets
\[
\begin{array}{rclcrclcrclc}
\{y_1,y_2\}_{1\eta}&=&\displaystyle \frac{\sin(\eta y_4)}{4\eta},& \{y_1,y_3\}_{1\eta}&=&\displaystyle \frac{\cos(\eta y_4) -1}{4\eta},&
\{y_2,y_3\}_{1\eta}&=& \displaystyle -\frac{y_1}{2}, \\[8pt]
\{z_1,z_2\}_{1\eta}&=&\displaystyle \frac{\sin(\eta z_4)}{4\eta},&
\{z_1,z_3\}_{1\eta}&=& \displaystyle \frac{\cos(\eta z_4) - 1}{4\eta},& 
\{z_2,z_3\}_{1\eta}&=& \displaystyle -\frac{z_1}{2}.
\end{array}
\]

If $pr_i: G_\eta \times G_\eta \to G_\eta$, with $i \in \{1, 2\}$, are the canonical projections then the Casimir functions for the Poisson brackets $\{\cdot, \cdot\}_{0\eta}$ and $\{\cdot, \cdot\}_{1\eta}$ are
\[
{\mathcal C}_{0\eta} \circ pr_1 = y_2^{2} + y_3^2, \qquad {\mathcal C}_{0\eta} \circ pr_2 = z_2^{2} + z_3^2, \qquad {\mathcal C}'_{0\eta} \circ pr_1 = y_4, \qquad {\mathcal C}'_{0\eta} \circ pr_2 = z_4,
\]
and  
\[
\begin{array}{rcl}
{\mathcal C}_{1\eta} \circ pr_1 &=& \displaystyle \frac{\sin(\eta y_4)}{\eta}y_3-\frac{\cos(\eta y_4)-1}{\eta}y_2-y_1^2, \\[8pt] 
{\mathcal C}_{1\eta} \circ pr_2 &= & \displaystyle \frac{\sin(\eta z_4)}{\eta}z_3-\frac{\cos(\eta z_4)-1}{\eta}z_2-z_1^2, \\[8pt] 
{\mathcal C}'_{1\eta} \circ pr_1 & = & y_4, \; \; \; \; \; {\mathcal C}'_{1\eta} \circ pr_2  =  z_4,
\end{array}
\]
respectively.

Now, we consider on $G_\eta \times G_\eta$ the coproduct of the Hamiltonian functions $H_{0\eta}$ and $H_{1\eta}$ on $G_\eta$, which we also denote by $H_{0\eta}$ and $H_{1\eta}$,
\[
\begin{array}{rcl}
H_{0\eta} : = {\mathcal C}_{1\eta} \circ \cdot_{\eta} & = & \displaystyle \left( \frac{1 - \cos(\eta(y_4 + z_4))}{\eta}\right) y_2 + \frac{\sin(\eta(y_4 + z_4))}{\eta} y_3 \\[8pt]
&& + \displaystyle \left(\frac{\cos(\eta y_4) - \cos (\eta z_4)}{\eta}\right) z_2 + \displaystyle  \left(\frac{\sin(\eta y_4) + \sin (\eta z_4)}{\eta}\right) z_3-(y_1+z_1)^2,
\end{array}
\]
and
\[
H_{1\eta} : = {\mathcal C}_{0\eta} \circ \cdot_{\eta} = y_2^2 + y_3^2 + z_2^2 + z_3^2 + 2(y_2 z_2 + y_3 z_3) \cos(\eta y_4)
+ 2(y_2 z_3 - y_3 z_2) \sin(\eta y_4). 
\]

The Hamiltonian system $(G_\eta \times G_\eta, \{\cdot, \cdot\}_{0\eta}, H_{0\eta})$ can be straightforwardly computed and reads
\[
\begin{array}{rcl}
\dot{y}_1 &=& \displaystyle \left(\frac{\sin(\eta(y_4 + z_4))}{\eta}\right) \frac{y_2}{2} + \displaystyle \left(\frac{\cos(\eta(y_4 + z_4)) - 1}{\eta}\right) \frac{y_3}{2}, \\[8pt]
\dot{y}_2 & = & -y_3 (y_1 + z_1), \; \; \; \; \; \dot{y}_3  =  y_2 (y_1 + z_1),  \; \; \; \; \; \dot{y}_4  =  0,\\[8pt]
\dot{z}_1 & = & \displaystyle \left(\frac{\sin(\eta y_4) + \sin (\eta z_4)}{2\eta}\right) z_2 + \displaystyle \left(\frac{\cos(\eta z_4) - \cos (\eta y_4)}{2\eta}\right) z_3, \\[8pt]
\dot{z}_2 & = & -z_3 (y_1 + z_1), \; \; \; \; \; \dot{z}_3  =  z_2 (y_1 + z_1),  \; \; \; \; \; \dot{z}_4  =  0,
\end{array}
\]
while $(G_\eta \times G_\eta, \{\cdot, \cdot\}_{1\eta}, H_{1\eta})$ gives rise to the dynamical system
\[
\begin{array}{rcl}
\dot{y}_1 &=& \displaystyle \left(\frac{\sin(\eta y_4)}{2\eta}\right) y_2 + \displaystyle \left(\frac{\cos (\eta y_4) -1}{2\eta}\right) y_3 + \displaystyle \left(\frac{\sin(\eta y_4)}{2\eta}\right) z_2 - \displaystyle \left(\frac{\cos (\eta y_4) -1}{2\eta}\right) z_3, \\[5pt]
 \dot{y}_2 & = & -y_1 (y_3 - z_2 \sin(\eta y_4) + z_3 \cos(\eta y_4)), \\[5pt] 
 \dot{y}_3 & = & y_1 (y_2 + z_2 \cos(\eta y_4) + z_3 \sin (\eta y_4)), \\[5pt]
 \dot{y}_4 & = & 0,\\[8pt]
\dot{z}_1 & = &\displaystyle \left(\frac{\sin(\eta (y_4 + z_4)) - \sin (\eta y_4)}{2\eta}\right) y_2 + \displaystyle \left(\frac{\cos(\eta (y_4 + z_4)) - \cos (\eta y_4)}{2\eta}\right) y_3 \\[8pt]
&& + \displaystyle \left(\frac{\sin(\eta z_4)}{2 \eta}\right) z_2 + \displaystyle \left(\frac{\cos(\eta z_4) -1}{2 \eta}\right) z_3 \\[8pt]
\dot{z}_2 & = & -z_1 (\sin(\eta y_4) y_2 + \cos(\eta y_4) y_3 + z_3), \\[5pt] 
\dot{z}_3 & = & z_1 (\cos(\eta y_4) y_2 - \sin(\eta y_4) y_3 + z_2), \\[5pt] 
\dot{z}_4 & = & 0.
\end{array}
\]

As we know, both systems are completely integrable and they admit reduction, via the multiplication $\cdot_\eta : G_\eta \times G_\eta \to G_\eta$, to the bi-Hamiltonian system $D_{\eta}$ on $G_\eta$ in Section \ref{bi-Ha-D-eta}. This last result becomes apparent if we consider the new coordinates $(x, z) = ((x_1, x_2, x_3, \linebreak x_4), (z_1, z_2, z_3, z_4))$ on $G_\eta \times G_\eta$, with $x = y \cdot_{\eta}  z$.

It is straightforward to prove that in these new coordinates, we have the following expressions 
for the Poisson structures $\Pi_{0\eta}$ and $\Pi_{1\eta}$ on $G_\eta \times G_\eta$:
\[
\begin{array}{rcl}
\Pi_{0\eta} (x, z) & = & \displaystyle -\frac{x_3}{2} \frac{\partial}{\partial x_1} \wedge \frac{\partial}{\partial x_2}  \displaystyle +\frac{x_2}{2} \frac{\partial}{\partial x_1} \wedge \frac{\partial}{\partial x_3} \displaystyle -\frac{z_3}{2} \frac{\partial}{\partial x_1} \wedge \frac{\partial}{\partial z_2} \displaystyle +\frac{z_2}{2} \frac{\partial}{\partial x_1} \wedge \frac{\partial}{\partial z_3} \\[8pt]
&& + \displaystyle \left(\cos(\eta(x_4 - z_4))\frac{z_3}{2} - \sin(\eta(x_4 - z_4)) \frac{z_2}{2}\right) \frac{\partial}{\partial x_2} \wedge \frac{\partial}{\partial z_1}  \displaystyle +\frac{z_2}{2} \frac{\partial}{\partial z_1} \wedge \frac{\partial}{\partial z_3} \\[8pt]
&& \displaystyle - \left(\cos(\eta(x_4 - z_4))\frac{z_2}{2} + \sin(\eta(x_4 - z_4)) \frac{z_3}{2}\right) \frac{\partial}{\partial x_3} \wedge \frac{\partial}{\partial z_1}  \displaystyle -\frac{z_3}{2} \frac{\partial}{\partial z_1} \wedge \frac{\partial}{\partial z_2},
\end{array}
\]
\[
\begin{array}{rcl}
\Pi_{1\eta}(x,z)&=&\displaystyle \frac{\sin(\eta x_4)}{4\eta} \frac{\partial}{\partial x_1}\wedge \frac{\partial}{\partial x_2}
+ \displaystyle \left(\frac{\cos(\eta x_4) -1}{4\eta}\right) \frac{\partial}{\partial x_1}\wedge \frac{\partial}{\partial x_3} -\displaystyle \frac{x_1}{2}\frac{\partial}{\partial x_2}\wedge \frac{\partial}{\partial x_3} \\[8pt] 
&& + \displaystyle \frac{\sin(\eta z_4)}{4\eta} \frac{\partial}{\partial x_1}\wedge \frac{\partial}{\partial z_2} + \displaystyle \left(\frac{\cos(\eta z_4) -1}{4 \eta}\right) \frac{\partial}{\partial x_1}\wedge \frac{\partial}{\partial z_3} -\frac{z_1}{2}\frac{\partial}{\partial z_2}\wedge \frac{\partial}{\partial z_3}
\\[8pt]
&&+\displaystyle \left(\frac{\sin(\eta(x_4 - z_4)) - \sin(\eta x_4)}{4 \eta}\right)\frac{\partial}{\partial x_2}\wedge \frac{\partial}{\partial z_1} + \displaystyle \frac{\sin(\eta z_4)}{4\eta} \frac{\partial}{\partial z_1}\wedge \frac{\partial}{\partial z_2} \\[8pt]
&& + \displaystyle \frac{z_1}{2} \sin(\eta(x_4 - z_4)) \frac{\partial}{\partial x_2}\wedge \frac{\partial}{\partial z_2} - \displaystyle \frac{z_1}{2} \cos(\eta(x_4 - z_4)) \frac{\partial}{\partial x_2}\wedge \frac{\partial}{\partial z_3} \\[8pt]
&& + \displaystyle \left(\frac{\cos(\eta(x_4 - z_4)) - \cos(\eta x_4)}{4 \eta}\right) \frac{\partial}{\partial x_3}\wedge \frac{\partial}{\partial z_1} + \displaystyle \frac{z_1}{2} \cos(\eta(x_4 - z_4)) \frac{\partial}{\partial x_3}\wedge \frac{\partial}{\partial z_2} \\[8pt]
&& + \displaystyle \frac{z_1}{2} \sin(\eta(x_4 - z_4)) \frac{\partial}{\partial x_3}\wedge \frac{\partial}{\partial z_3} + \displaystyle \left(\frac{\cos(\eta z_4) -1}{4\eta}\right) \frac{\partial}{\partial z_1}\wedge \frac{\partial}{\partial z_3}. 
\end{array}
\]

The Casimir functions of the Poisson bracket $\{\cdot, \cdot\}_{0\eta}$ are
\[
\begin{array}{rcl}
{\mathcal C}_{0\eta} \circ pr_1 & = & x_2^2 + x_3^2 + z_2^2 + z_3^2 - 2(x_2 z_2 + x_3 z_3) \cos(\eta(x_4 - z_4)) \\[6pt] 
&& + 2(x_3 z_2 - x_2 z_3) \sin(\eta(x_4 - z_4)), \\[6pt] 
{\mathcal C}_{0\eta} \circ pr_2 & = & z_2^2 + z_3^2, \; \; \; \; {\mathcal C}'_{0\eta} \circ pr_1 = x_4 - z_4, \; \; \; \;  {\mathcal C}'_{0\eta} \circ pr_2 = z_4,
\end{array}
\]
and for $\{\cdot, \cdot\}_{1\eta}$ we have
\[
\begin{array}{rcl}
{\mathcal C}_{1\eta} \circ pr_1 &=& \displaystyle \frac{\sin(\eta (x_4 - z_4))}{\eta}(x_3 -z_3) - \left(\frac{\cos(\eta (x_4 - z_4)-1}{\eta}\right) (x_2 + z_2) - (x_1 - z_1)^2, \\[8pt] 
{\mathcal C}_{1\eta} \circ pr_2 &= & \displaystyle \frac{\sin(\eta z_4)}{\eta}z_3-\left(\frac{\cos(\eta z_4)-1}{\eta}\right)z_2-z_1^2, \\[8pt] 
{\mathcal C}'_{1\eta} \circ pr_1 & = & x_4 - z_4, \; \; \; \; \; {\mathcal C}'_{1\eta} \circ pr_2  =  z_4.
\end{array}
\]
On the other hand, the Hamiltonian functions $H_{0\eta}$ and $H_{1\eta}$ read
\[
H_{0\eta} = \displaystyle \frac{\sin (\eta x_4)}{\eta} x_3 - \left(\frac{\cos(\eta x_4) -1}{\eta}\right)x_2 - x_1^2, \qquad H_{1\eta} = x_2^2 + x_3^2.
\]

In these coordinates, the two completely integrable Hamiltonian systems $(G_\eta \times G_\eta, \{\cdot, \cdot\}_{0\eta}, H_{0\eta})$ and $(G_\eta \times G_\eta, \{\cdot, \cdot\}_{1\eta}, H_{1\eta})$ become
\begin{equation}\label{coupledH_0-eta}
\begin{array}{rcl}
\dot{x_1}&=& \displaystyle \frac{1}{2} \frac{\sin(\eta x_4)}{\eta}x_2 + \left(\frac{\cos(\eta x_4) -1}{2\eta}\right) x_3, \; \; \; \dot{x_2}= -x_1x_3, \; \; \;  \dot{x_3} =  x_1x_2, \; \; \; \dot{x_4} =  0,   \\[8pt]
\dot{z_1}&=& \displaystyle \left( \frac{\sin(\eta(x_4 - z_4)) + \sin (\eta z_4)}{2\eta}\right) z_2 + \displaystyle \left( \frac{\cos(\eta(z_4)) - \cos (\eta (x_4 - z_4)}{2\eta}\right) z_3, \\[8pt] 
\dot{z_2}&=& - x_1z_3, \; \; \;  \dot{z_3} =  x_1z_2, \; \; \; \dot{z_4} = 0,
\end{array}
\end{equation}
and
\begin{equation}\label{coupledH_1-eta}
\begin{array}{rcl}
\dot{x_1}&=& \displaystyle \frac{1}{2} \frac{\sin(\eta x_4)}{\eta}x_2 + \left(\frac{\cos(\eta x_4) -1}{2\eta}\right) x_3, \; \; \; \dot{x_2}= -x_1x_3, \; \; \;  \dot{x_3} =  x_1x_2, \; \; \; \dot{x_4} =  0,    \\[8pt]
\dot{z_1}&=& \displaystyle \left( \sin (\eta x_4) - \frac{\sin(\eta(x_4 - z_4))}{2\eta}\right) x_2 + \displaystyle \left( \frac{\cos(\eta(x_4)) - \cos (\eta (x_4 - z_4)}{2\eta}\right) x_3, \\[8pt]
\dot{z_2} & = & -z_1(x_2 \sin(\eta(x_4 - z_4)) + x_3 \cos(\eta (x_4 - z_4))), \\[6pt]
\dot{z_3} & = & z_1(x_2 \cos(\eta(x_4 - z_4)) - x_3 \sin(\eta (x_4 - z_4))), \\[6pt]
\dot{z_4} & = & 0, 
\end{array}
\end{equation}
respectively.
Again, the multiplication $\cdot_{\eta} : G_\eta \times G_\eta \to G_\eta$ leads to the projection
\begin{equation}\label{multiplication}
\cdot_{\eta}((x_1, x_2, x_3, x_4), (z_1, z_2, z_3, z_4)) = (x_1, x_2, x_3, x_4).
\end{equation}
So, by recalling (\ref{bi-Hamil-non-abelian}), (\ref{coupledH_0-eta}), (\ref{coupledH_1-eta}) and (\ref{multiplication}), we directly deduce that the two completely integrable Hamiltonian systems $(G_\eta \times G_\eta, \{\cdot, \cdot\}_{0\eta}, H_{0\eta})$ and $(G_\eta \times G_\eta, \{\cdot, \cdot\}_{1\eta}, H_{1\eta})$ admit reduction, via the multiplication $\cdot_{\eta} : G_\eta \times G_\eta \to G_\eta$, to the bi-Hamiltonian system $D_{\eta}$ on $G_\eta$ considered in Section \ref{bi-Ha-D-eta}.


\section{Another example: an Euler top system}\label{Euler-top-system}
In this section, we will discuss another example: an Euler top system. We will follow the same steps as in the Lorenz system.  So, first of all, we will present the dynamical system and its bi-Hamiltonian structure.

\subsection{The system $D$ and its bi-Hamiltonian structure}\label{D-bi-Ha-st}

We consider the following completely integrable system $D$ on $\mathbb{R}^3$ 
\begin{equation}\label{system 3}
\begin{array}{rcl}
\dot{x_1}&=&x_2^2-x_3^2,\\
\dot{x_2}&=& x_1(2x_3-x_2),\\
\dot{x_3}&=&x_1(x_3-2x_2).
\end{array}
\end{equation}
This system is equivalent to a particular case of the ${so}(3)$ Euler top, which is a well-known three dimensional bi-Hamiltonian system (see \cite{GuNu}) belonging to the realm of classical mechanics \cite{Rey}.

In fact, the previous system is bi-Hamiltonian with respect to the Lie-Poisson structures $\{\cdot, \cdot \}_{0}$ and $\{\cdot, \cdot \}_{1}$ in $\mathbb{R}^3$ which are characterized by 
\begin{equation}\label{Poissen-algebra-Euler}
\begin{array}{rclcrclcrcl}
\{x_1,x_2\}_{0} &= & \displaystyle -x_3,\quad & \{x_1,x_3\}_{0} &= & \displaystyle  x_2,\quad & \{x_2,x_3\}_0& = &-x_1,\\[5pt]
\{x_1,x_2\}_{1}& = & \displaystyle -x_2,\quad & \{x_1,x_3\}_{1} & = &\displaystyle x_3,\quad & \{x_2,x_3\}_1 &= &-2x_1.
\end{array}
\end{equation}
The Casimirs for these Poisson structures are 
\[
 \mathcal C_{0}=-\frac{1}{2}(x_1^2+x_2^2+x_3^2)\quad  \mbox{ and } \quad  \mathcal C_{1}= x_1^2+x_2x_3.
\]
It is straightforward to prove that the Hamiltonian systems $(\{\cdot, \cdot\}_0, H_0: = {\mathcal C}_1)$ and $(\{\cdot, \cdot\}_1, H_1 : = {\mathcal C}_0)$ coincide with the system $D$. 

The real Lie algebras corresponding to $\{.,.\}_0$ and $\{.,.\}_1$ are $so(3)$ and $sl(2;R)$, respectively. The structure equations for $so(3)$ and $sl(2, \mathbb{R})$ are
\[
[X_1,X_2]_{0}=-X_3,\quad [X_1,X_3]_{0}=X_2,\quad [X_2,X_3]_{0}=-X_1,
\]
\[
[X_1,X_2]_{1}=-X_2,\quad [X_1,X_3]_{1}=X_3,\quad [X_2,X_3]_{1}=-2X_1.
\]
We consider the family of compatible Lie-Poisson structures: $$\quad\{\cdot,\cdot\}_{\alpha}=(1-\alpha)\{\cdot,\cdot\}_{0}+\alpha\{\cdot,\cdot\}_{1},\quad \mbox{ with } \alpha \in \mathbb{R}$$ 
which are characterized by
\begin{equation}
\begin{array}{rclcrcl}\label{linear.p}
\{x_1,x_2\}_{\alpha}&=&(\alpha-1)x_3-\alpha x_2,\quad \{x_1,x_3\}_{\alpha}&=&(1-\alpha)x_2+\alpha x_3,\\ \{x_2,x_3\}_{\alpha}&=&-(1+\alpha)x_1.
\end{array}
\end{equation}
If $\{X_1, X_2, X_3\}$ is the canonical basis of $\mathbb{R}^3$ then the corresponding Lie bracket $[\cdot, \cdot]_{\alpha}$ on $\mathbb{R}^3$ is given by
\[
\begin{array}{rclcrcl}
[X_1,X_2]_{\alpha}&=&(\alpha-1)X_3-\alpha X_2,\quad [X_1,X_3]{\alpha}&=&(1-\alpha)X_2+\alpha X_3,\\
\left[ X_1,X_3\right] _{\alpha}&=&-(1+\alpha)X_1.
\end{array}
\]
So, we have a family of Lie algebras $({\frak g}_{\alpha}, [\cdot, \cdot]_{\alpha})$.

\subsection{Construction of the bi-Hamiltonian system $D_{\eta}$}\label{Con-bi-Ha-sy}
First of all, we will consider a family of non-trivial admissible $1$-cocycles for the previous Lie algebras $({\frak g}_{\alpha}, [\cdot, \cdot]_{\alpha})$ given by:
\[
\psi_\eta(X_1)=0,  \; \;
\psi_\eta(X_2)=\eta X_2\wedge X_1,  \; \;
\psi_\eta(X_3)=\eta X_3\wedge X_1.
\]

Therefore, we have an $\eta$-parametric family of Lie bialgebras $(\mathfrak{g}_\alpha,\psi_{\eta})$. The Lie bracket $[\cdot, \cdot]_{\eta}^*$ on $(\mathbb{R}^3)^* \simeq \mathbb{R}^3$ obtained from the dual cocommutator map is:
\[
[X^1,X^2]_{\eta}^*=-\eta X^2,\quad [X^1,X^3]_{\eta}^*=-\eta X^3, \quad
[X^2,X^3]_{\eta}^*=0.
\]
So, $(\mathbb{R}^3, [\cdot, \cdot]_\eta^*)$ is just the so-called book Lie algebra.

Now, let $G_{\eta}$ be a connected simply-connected Lie group with Lie algebra $(\mathbb{R}^3, [\cdot, \cdot]_{\eta}^*)$. Then, one may prove that $G_{\eta}$ is diffeomorphic to $\mathbb{R}^3$ and the multiplication of two elements $g = (x_1, x_2, x_3)$ and $g'=(x_1', x_2', x_3')$ of $\mathbb{R}^3$ is given by 
$$
g._\eta g' =
( x_1+x'_1,x_2+x'_2e^{-\eta x_1},x_3+x'_3e^{-\eta x_1}).
$$
A basis $\{\lvec X^1,\lvec X^2,\lvec X^3\}$ (resp., $\{\rvec X^1,\rvec X^2,\rvec X^3\}$) of left-invariant (resp., right-invariant) vector fields is
\[
\{\frac{\partial}{\partial x_1}, e^{-\eta x_1}\frac{\partial}{\partial x_2},e^{-\eta x_1}\frac{\partial}{\partial x_3}\}
\]
(resp., $\displaystyle \{\frac{\partial}{\partial x_1}-\eta x_2  \frac{\partial}{\partial x_2}-\eta x_3  \frac{\partial}{\partial x_3}, \frac{\partial}{\partial x_2}, \frac{\partial}{\partial x_3}\}$).

The adjoint action $Ad:G_{\eta}\times{\mathfrak {g}}_{\eta} \longrightarrow {\mathfrak {g}}_{\eta}$ for Lie group $G_{\eta}$ is as follows:
\[
Ad_{g}(X^1)= \eta(x_2X^2+x_3X^3)+X^1,\quad
Ad_{g}(X^2)= e^{-\eta x_1}X^2, \quad
Ad_{g}(X^3)=  e^{-\eta x_1} X^3.
\]
Next, as in the case of the Lorenz system, we will introduce a Poisson-Lie group structure on $G_\eta$.

As we know, a family of non-trivial admissible 1-cocycles for the Lie algebra $(\mathbb{R}^3, [\cdot, \cdot]_{\eta}^*)$ is 
\[
\begin{array}{rclcrcl}
\psi_{\alpha}(X^1) & = &-(1+\alpha)X^1\wedge X^3,\ &  \psi_{\alpha}(X^2)& = & -\alpha X^1\wedge X^2+(1-\alpha)X^1\wedge X^3, \\[7pt]
 \psi_{\alpha}(X^3) &= & (\alpha-1)X^1\wedge X^2+\alpha X^1\wedge X^3.\\
\end{array}
\]
Applying the same process as in the previous example (see Section \ref{Po-Li-gr-st}), 
we deduce that the corresponding compatible
multiplicative Poisson brackets $\{\cdot, \cdot \}_{\alpha \eta}$ on $G_{\eta}$ are characterized by 
\[
\begin{array}{rcl}
\{x_1,x_2\}_{\alpha \eta}& = & \displaystyle (\alpha -1)x_3-\alpha x_2,\\[5pt]
 \{x_1,x_3\}_{\alpha \eta}& = &(1-\alpha)x_2+\alpha x_3,\\[5pt]
 \{x_2,x_3\}_{\alpha \eta}& = & \eta(\alpha-1)\displaystyle (\frac{x_2^2}{2}+\frac{x_3^2}{2})-\eta \alpha x_2x_3+(1+\alpha)\frac{e^{-2\eta x_1}-1}{2\eta}.
\end{array}
\]
As expected, $\lim _{\eta\rightarrow 0}\{.,.\}_{\alpha \eta}=\{.,.\}_{\alpha}$ and this means that we have a $\eta$-deformation of the Lie-Poisson bracket (\ref{linear.p}).

In the cases when $\alpha=0$ and $\alpha=1$, the $\eta$-deformations $\{\cdot, \cdot\}_{0\eta}$ and $\{\cdot, \cdot\}_{1\eta}$ of the Lie-Poisson brackets $\{\cdot, \cdot\}_0$ and $\{\cdot, \cdot\}_1$ have the form
\begin{equation}\label{Poisson-Lie0-Euler}
\begin{array}{rclcrcl}
\{x_1,x_2\}_{0\eta}  &= & -x_3,&
\{x_1,x_3\}_{0\eta}& =& x_2, \\
 \{x_2,x_3\}_{0\eta}& =&\displaystyle -\eta(\frac{x_2^2}{2}+\frac{x_3^2}{2})+ \displaystyle\frac{e^{-2\eta x_1}-1}{2\eta} .\\
 \end{array}
 \end{equation}
 \begin{equation}\label{Poisson-Lie1-Euler}
 \begin{array}{rclcrcl}
\{x_1,x_2\}_{1\eta}&= &-x_2, &
\{x_1,x_3\}_{1 \eta}& = & x_3,\\
\{x_2,x_3\}_{1 \eta}& = &- \eta  x_2x_3+ \displaystyle\frac{e^{-2\eta x_1}-1}{\eta}.
\end{array}
 \end{equation}
Casimir functions for the previous two multiplicative Poisson structures are:
 \[
 \begin{array}{rcl}
  \mathcal C_{0\eta}&=& \displaystyle -e^{\eta x_1} \frac{x_2^2+x_3^2}{2}-\displaystyle \frac{e^{\eta x_1}+e^{-\eta x_1}-2}{2\eta^2}, \; \;\\
\mathcal C_{1\eta}&=& e^{\eta x_1}(x_2x_3)+\displaystyle \frac{e^{\eta x_1}+e^{-\eta x_1}-2}{\eta^2}.
\end{array}
 \]
 It is clear that $\lim_{\eta \rightarrow 0}\mathcal C_{0\eta}=\mathcal C_{0}$ and $\lim_{\eta \rightarrow 0}\mathcal C_{1\eta}=\mathcal C_{1}$. If we denote the Casimir functions $\mathcal C_{0\eta}$ and $\mathcal C_{1\eta}$ by $H_{1\eta}$ and $H_{0\eta}$, respectively, then the dynamical systems associated with the Hamiltonian systems $(\mathbb{R}^3, \{\cdot, \cdot \}_{0\eta}, H_{0\eta})$ and  $(\mathbb{R}^3, \{\cdot, \cdot \}_{1\eta}, H_{1\eta})$ coincide. In other words, the dynamical system on the Lie group $G_{\eta}$
  \begin{equation}\label{Euler.deformation}
\begin{array}{rcl}
\dot{x_1}&=& e^{\eta x_1}(x_2^2-x_3^2),\\
\dot{x_2}&=& \eta  e^{\eta x_1} x_2  x_3^2-\displaystyle\frac{1}{2}\eta e^{\eta x_1}(x_2^3-x_2 x_3^2)+\displaystyle  \frac{e^{\eta x_1}-e^{-\eta x_1}}{2 \eta}(2 x_3-x_2),\\
 \dot{x_3}&=& -\eta  e^{\eta x_1} x_2^2  x_3+\displaystyle\frac{1}{2}\eta e^{\eta x_1}(x_2^2 x_3 +x_3^3)+\displaystyle  \frac{e^{\eta x_1}-e^{-\eta x_1}}{2 \eta}(x_3-2 x_2).
\end{array}
 \end{equation}
is bi-Hamiltonian with respect to the compatible multiplicative Poisson structures $\{\cdot, \cdot\}_{0\eta}$ and $\{\cdot, \cdot\}_{1\eta}$.

From the previous considerations, we also deduce that the bi-Hamiltonian system is completely integrable.

Finally, as we expected, the limit when $\eta$ approaches to zero of (\ref{Euler.deformation}) is just the bi-Hamiltonian system (\ref{system 3}).  

\subsection{Deformed coupled Euler top systems}
Denote by $\Pi_{0\eta}$ and $\Pi_{1\eta}$ the multiplicative Poisson structures on the Lie group $G_\eta$ associated with the Poisson brackets $\{\cdot, \cdot\}_{0\eta}$ and $\{\cdot, \cdot\}_{1\eta}$ given by (\ref{Poisson-Lie0-Euler}) and (\ref{Poisson-Lie1-Euler}), respectively, and by $(y, z) = ((y_1, y_2, y_3), (z_1, z_2, z_3))$ the standard coordinates on $G_\eta \times G_\eta \simeq \mathbb{R}^3 \times \mathbb{R}^3$.

Then, we can consider the multiplicative Poisson structures on  $G_\eta \times G_\eta$:

\noindent $\bullet$
The product of $\Pi_{0\eta}$ with itself, that is, the Poisson bracket on $G_\eta \times G_\eta$, which we also denote by $\{\cdot, \cdot\}_{0\eta}$, characterized by
\[
\begin{array}{rclcrcl}
\{y_1,y_2\}_{0\eta} & = & -y_3, &
\{y_1,y_3\}_{0\eta} & = & y_2, \\
 \{y_2,y_3\}_{0\eta} & = & \displaystyle -\eta(\frac{y_2^2}{2}+\frac{y_3^2}{2})+ \displaystyle\frac{e^{-2\eta y_1}-1}{2\eta} , \\ \{z_1,z_2\}_{0\eta} & = &-z_3,&  \{z_1,z_3\}_{0\eta} & = & z_2,\\ \{z_2,z_3\}_{0\eta} & = &\displaystyle -\eta(\frac{z_2^2}{2}+\frac{z_3^2}{2})+ \displaystyle\frac{e^{-2\eta z_1}-1}{2\eta}. \\
\end{array}
\]

\noindent $\bullet$
The product of $\Pi_{1\eta}$ with itself, that is, the Poisson bracket on $G_\eta \times G_\eta$, which we also denote by $\{\cdot, \cdot\}_{1\eta}$, characterized by
\[
\begin{array}{rclcrcl}
\{y_1,y_2\}_{1\eta} & = & -y_2, &
\{y_1,y_3\}_{1\eta} & = & y_3, \\
 \{y_2,y_3\}_{1\eta} & = & - \eta  y_2y_3+ \displaystyle\frac{e^{-2\eta y_1}-1}{\eta},&
\\
 \{z_1,z_2\}_{1\eta} & = & -z_2,&    \{z_1,z_3\}_{1\eta} & = & z_3,
   \\
   \{z_2,z_3\}_{1\eta} & = & - \eta  z_2z_3+ \displaystyle\frac{e^{-2\eta z_1}-1}{\eta}.\\
\end{array}
\]

If $pr_i: G_\eta \times G_\eta \to G_\eta$, with $i \in \{1, 2\}$, are the canonical projections then the Casimir functions for the Poisson brackets $\{\cdot, \cdot\}_{0\eta}$ and $\{\cdot, \cdot\}_{1\eta}$ are

\[
\begin{array}{rcl}
{\mathcal C}_{0\eta} \circ pr_1 &=&  -\displaystyle e^{\eta y_1} \frac{y_2^2+y_3^2}{2}-\displaystyle \frac{e^{\eta y_1}+e^{-\eta y_1}-2}{2\eta^2}, \\[8pt] 
{\mathcal C}_{0\eta} \circ pr_2 &= & - \displaystyle e^{\eta z_1} \frac{z_2^2+z_3^2}{2}-\displaystyle \frac{e^{\eta z_1}+e^{-\eta z_1}-2}{2\eta^2}, \\
\end{array}
\]
and  
\[
\begin{array}{rcl}
{\mathcal C}_{1\eta} \circ pr_1 &=& e^{\eta y_1}(y_2y_3)+\displaystyle \frac{e^{\eta y_1}+e^{-\eta y_1}-2}{\eta^2}, \\[8pt] 
{\mathcal C}_{1\eta} \circ pr_2 &= & e^{\eta z_1}(z_2z_3)+\displaystyle \frac{e^{\eta z_1}+e^{-\eta z_1}-2}{\eta^2}, \\
\end{array}
\]
respectively.

Now, we consider on $G_\eta \times G_\eta$ the coproduct of the Hamiltonian functions $H_{0\eta}$ and $H_{1\eta}$ on $G_\eta$, which we also denote by $H_{0\eta}$ and $H_{1\eta}$,
\[
H_{0\eta} : = {\mathcal C}_{1\eta} \circ \cdot_{\eta} =e^{\eta z_1}(y_2z_3+y_3z_2)+e^{\eta(y_1+z_1)} y_2y_3+e^{\eta(z_1-y_1)}z_2z_3 +\frac{e^{\eta (y_1+z_1)}+e^{-\eta (y_1+z_1)}-2}{\eta^2},
\]
and
\[
\begin{array}{rcl}
H_{1\eta} : = {\mathcal C}_{0\eta} \circ \cdot_{\eta}&=&\displaystyle- e^{\eta z_1}(y_2z_2+y_3z_3)- e^{\eta (y_1+z_1)}\left( \frac{y_2^2}{2}+\frac {y_3^2}{2}\right)- e^{\eta (z_1-y_1)}\left( \frac{z_2^2}{2}+\frac {z_3^2}{2}\right)\\[8pt]
&&-\displaystyle \frac{e^{\eta (y_1+z_1)}+e^{-\eta (y_1+z_1)}-2}{2\eta^2}.
\end{array} 
\]
Then, the Hamiltonian systems $(G_\eta \times G_\eta, \{\cdot, \cdot\}_{0\eta}, H_{0\eta})$ and $(G_\eta \times G_\eta, \{\cdot, \cdot\}_{1\eta}, H_{1\eta})$ are given by
\[
\begin{array}{rcl}
\dot{y}_1 &=& \displaystyle-y_3\left (y_3e^{\eta(y_1+z_1)}+z_3e^{\eta z_1}\right )+y_2\left(y_2e^{\eta(y_1+z_1)}+z_2e^{\eta z_1}\right),\\[6pt]
\dot{y}_2 & = &\displaystyle y_3\left( \eta y_2y_3e^{\eta(y_1+z_1)}-\eta z_2z_3e^{\eta(z_1-y_1)}+\frac{e^{\eta(y_1+z_1)}-e^{-\eta(y_1+z_1)}}{\eta}\right)\\[6pt]
&&+\displaystyle \left(y_2e^{\eta(y_1+z_1)}+z_2e^{\eta z_1}\right)\left( \frac{e^{-2\eta y_1}-1}{2\eta}-\eta(\frac{y_2^2}{2}+\frac{y_3^2}{2})\right),\\[6pt]
\dot{y}_3 & = &\displaystyle -y_2\left( \eta y_2y_3e^{\eta(y_1+z_1)}-\eta z_2z_3e^{\eta(z_1-y_1)}+\frac{e^{\eta(y_1+z_1)}-e^{-\eta(y_1+z_1)}}{\eta}\right)\\[6pt]
&&+\displaystyle \left(y_3e^{\eta(y_1+z_1)}+z_3e^{\eta z_1}\right)\left( \eta(\frac{y_2^2}{2}+\frac{y_3^2}{2})-\frac{e^{-2\eta y_1}-1}{2\eta}\right),\\[6pt]
\dot{z}_1 & = & \displaystyle-z_3\left (y_3e^{\eta z_1}+z_3e^{\eta (z_1-y_1)}\right )+z_2\left (y_2e^{\eta z_1}+z_2e^{\eta (z_1-y_1)}\right),\\[6pt]
\dot{z}_2 & = & \displaystyle z_3\left( \eta y_2y_3e^{\eta(y_1+z_1)}+\eta(y_2z_3+y_3z_2)e^{\eta z_1}+\eta z_2z_3e^{\eta(z_1-y_1)}+\frac{e^{\eta(y_1+z_1)}-e^{-\eta(y_1+z_1)}}{\eta}\right)\\[6pt]
&&+\displaystyle \left(y_2e^{\eta z_1}+z_2e^{\eta (z_1-y_1)}\right)\left( \frac{e^{-2\eta z_1}-1}{2\eta}-\eta(\frac{z_2^2}{2}+\frac{z_3^2}{2})\right),\\[6pt]
\dot{z}_3 & = & \displaystyle -z_2\left( \eta y_2y_3e^{\eta(y_1+z_1)}+\eta(y_2z_3+y_3z_2)e^{\eta z_1}+\eta z_2z_3e^{\eta(z_1-y_1)}+\frac{e^{\eta(y_1+z_1)}-e^{-\eta(y_1+z_1)}}{\eta}\right)\\[6pt]
&&+\displaystyle \left(y_3e^{\eta z_1}+z_3e^{\eta (z_1-y_1)}\right)\left( \eta(\frac{z_2^2}{2}+\frac{z_3^2}{2})-\frac{e^{-2\eta z_1}-1}{2\eta}\right),\\
\end{array}
\]
and 
\[
\begin{array}{rcl}
\dot{y}_1 &=& \displaystyle y_2\left (y_2e^{\eta(y_1+z_1)}+z_2e^{\eta z_1}\right )-y_3\left(y_3e^{\eta(y_1+z_1)}+z_3e^{\eta z_1}\right),\\[6pt]
 \dot{y}_2 & = & -\displaystyle y_2\left( \eta( \frac{y_2^2}{2}+\frac{y_3^2}{2})e^{\eta(y_1+z_1)}-\eta ( \frac{z_2^2}{2}+\frac{z_3^2}{2})e^{\eta(z_1-y_1)}+\frac{e^{\eta(y_1+z_1)}-e^{-\eta(y_1+z_1)}}{2\eta}\right)\\[6pt]
 &&-\displaystyle \left(y_3e^{\eta(y_1+z_1)}+z_3e^{\eta z_1}\right)\left( \frac{e^{-2\eta y_1}-1}{\eta}-\eta y_2y_3\right),\\[6pt]
 \dot{y}_3 & = & \displaystyle y_3\left( \eta( \frac{y_2^2}{2}+\frac{y_3^2}{2})e^{\eta(y_1+z_1)}-\eta ( \frac{z_2^2}{2}+\frac{z_3^2}{2})e^{\eta(z_1-y_1)}+\frac{e^{\eta(y_1+z_1)}-e^{-\eta(y_1+z_1)}}{2\eta}\right)\\[6pt]
  &&-\displaystyle \left(y_2e^{\eta(y_1+z_1)}+z_2e^{\eta z_1}\right)\left( \eta y_2y_3-\frac{e^{-2\eta y_1}-1}{\eta}\right),\\[6pt]
\dot{z}_1 & = &\displaystyle z_2\left (y_2e^{\eta z_1}+z_2e^{\eta(z_1-y_1)}\right )-z_3\left(y_3e^{\eta z_1}+z_3e^{\eta(z_1-y_1)}\right),\\[6pt]
\dot{z}_2 & = & \displaystyle- z_2\left( \eta( \frac{y_2^2}{2}+\frac{y_3^2}{2})e^{\eta(y_1+z_1)}+\eta ( \frac{z_2^2}{2}+\frac{z_3^2}{2})e^{\eta(z_1-y_1)}+\frac{e^{\eta(y_1+z_1)}-e^{-\eta(y_1+z_1)}}{2\eta}\right)\\[6pt]
 &&-\eta e^{\eta z_1}(y_2z_2+y_3z_3)z_2-\displaystyle \left(y_3e^{\eta z_1}+z_3e^{\eta(z_1-y_1)}\right)\left( \frac{e^{-2\eta z_1}-1}{\eta}-\eta z_2z_3\right),\\[8pt]
\dot{z}_3 & = & \displaystyle z_3\left( \eta( \frac{y_2^2}{2}+\frac{y_3^2}{2})e^{\eta(y_1+z_1)}+\eta ( \frac{z_2^2}{2}+\frac{z_3^2}{2})e^{\eta(z_1-y_1)}+\frac{e^{\eta(y_1+z_1)}-e^{-\eta(y_1+z_1)}}{2\eta}\right)\\[8pt]
 &&+\eta e^{\eta z_1}(y_2z_2+y_3z_3)z_3-\displaystyle \left(y_2e^{\eta z_1}+z_2e^{\eta(z_1-y_1)}\right)\left(\eta z_2z_3- \frac{e^{-2\eta z_1}-1}{\eta}\right),\\[6pt]
\end{array}
\]
respectively.

As we know, these systems are completely integrable and they admit reduction, via the multiplication $\cdot_\eta : G_\eta \times G_\eta \to G_\eta$, to the bi-Hamiltonian system on $G_\eta$ in Section \ref{Con-bi-Ha-sy}.

This last result can be straightforwardly shown by considering the new coordinates $(x, z) = ((x_1, x_2, x_3), (z_1, z_2, z_3))$ on $G_\eta \times G_\eta$, with $x = y \cdot_{\eta}  z$. In these new coordinates, the Poisson structures $\Pi_{0\eta}$ and $\Pi_{1\eta}$ on $G_\eta \times G_\eta$ are given by
\[
\begin{array}{rcl}
\Pi_{0\eta}(x,z)&=&\displaystyle-x_3\frac{\partial}{\partial x_1}\wedge \frac{\partial}{\partial x_2}+x_2\frac{\partial}{\partial x_1}\wedge \frac{\partial}{\partial x_3}+  \displaystyle \left(\frac{e^{-2\eta x_1}-1}{2\eta}-\eta(\frac{x_2^2}{2}+\frac{x_3^2}{2})\right)\frac{\partial}{\partial x_2}\wedge \frac{\partial}{\partial x_3}\\
&&\displaystyle -z_3\frac{\partial}{\partial x_1}\wedge \frac{\partial}{\partial z_2}+z_2\frac{\partial}{\partial x_1}\wedge \frac{\partial}{\partial z_3}+
\displaystyle z_3e^{-\eta (x_1-z_1)}\frac{\partial}{\partial x_2}\wedge \frac{\partial}{\partial z_1}  
-\displaystyle z_2e^{-\eta (x_1-z_1)}\frac{\partial}{\partial x_3}\wedge \frac{\partial}{\partial z_1}  
\\[8pt]
&&+\displaystyle \left(\frac{e^{-\eta (x_1+z_1)}-e^{-\eta (x_1-z_1)}}{2\eta}-\eta e^{-\eta (x_1-z_1)}(\frac{z_2^2}{2}+\frac{z_3^2}{2})\right)\frac{\partial}{\partial x_2}\wedge \frac{\partial}{\partial z_3}\\[8pt]
&&+\displaystyle \left(\frac{e^{-\eta (x_1-z_1)}-e^{-\eta (x_1+z_1)}}{2\eta}+\eta e^{-\eta (x_1-z_1)}(\frac{z_2^2}{2}+\frac{z_3^2}{2})\right)\frac{\partial}{\partial x_3}\wedge \frac{\partial}{\partial z_2}\\
&&\displaystyle-z_3\frac{\partial}{\partial z_1}\wedge \frac{\partial}{\partial z_2}+z_2\frac{\partial}{\partial z_1}\wedge \frac{\partial}{\partial z_3}+  \displaystyle \left(\frac{e^{-2\eta z_1}-1}{2\eta}-\eta(\frac{z_2^2}{2}+\frac{z_3^2}{2})\right)\frac{\partial}{\partial z_2}\wedge \frac{\partial}{\partial z_3},\\
\end{array}
\]
and
\[
\begin{array}{rcl}
\Pi_{1\eta}(x,z)&=&\displaystyle-x_2\frac{\partial}{\partial x_1}\wedge \frac{\partial}{\partial x_2}+x_3\frac{\partial}{\partial x_1}\wedge \frac{\partial}{\partial x_3}+  \displaystyle \left(\frac{e^{-2\eta x_1}-1}{\eta}-\eta x_2x_3\right)\frac{\partial}{\partial x_2}\wedge \frac{\partial}{\partial x_3}\\
&&\displaystyle -z_2\frac{\partial}{\partial x_1}\wedge \frac{\partial}{\partial z_2}+z_3\frac{\partial}{\partial x_1}\wedge \frac{\partial}{\partial z_3}+
\displaystyle 
 z_2e^{-\eta (x_1-z_1)}\frac{\partial}{\partial x_2}\wedge \frac{\partial}{\partial z_1}  
-\displaystyle z_3 e^{-\eta (x_1-z_1)}\frac{\partial}{\partial x_3}\wedge \frac{\partial}{\partial z_1}  
\\[8pt]
&&+\displaystyle \left(\frac{e^{-\eta (x_1+z_1)}-e^{-\eta (x_1-z_1)}}{\eta}-\eta e^{-\eta (x_1-z_1)} z_2z_3\right)\frac{\partial}{\partial x_2}\wedge \frac{\partial}{\partial z_3}\\[8pt]
&&+\displaystyle \left(\frac{e^{-\eta (x_1-z_1)}-e^{-\eta (x_1+z_1)}}{\eta}+\eta e^{-\eta (x_1-z_1)} z_2z_3\right)\frac{\partial}{\partial x_3}\wedge \frac{\partial}{\partial z_2}\\
&&\displaystyle-z_2\frac{\partial}{\partial z_1}\wedge \frac{\partial}{\partial z_2}+z_3\frac{\partial}{\partial z_1}\wedge \frac{\partial}{\partial z_3}+  \displaystyle \left(\frac{e^{-2\eta z_1}-1}{\eta}-\eta z_2z_3\right)\frac{\partial}{\partial z_2}\wedge \frac{\partial}{\partial z_3}.\\
\end{array}
\]
The Casimir functions of the Poisson brackets $\{\cdot, \cdot\}_{0\eta}$ and $\{\cdot, \cdot\}_{1\eta}$ are:
\[
\begin{array}{rcl}
{\mathcal C}_{0\eta} \circ pr_1 & = & \displaystyle -\frac{1}{2}e^{\eta(x_1-z_1)}(x_2^2 + x_3^2)\displaystyle -\frac{1}{2}e^{-\eta(x_1-z_1)}(z_2^2 + z_3^2)+(x_2z_2+x_3z_3)\\[6pt] 
&&-\displaystyle \frac{e^{\eta (x_1-z_1)}+e^{-\eta (x_1-z_1)}-2}{2\eta^2} ,\\[8pt] 
{\mathcal C}_{0\eta} \circ pr_2 & = & \displaystyle -e^{\eta z_1} \frac{z_2^2+z_3^2}{2}-\displaystyle \frac{e^{\eta z_1}+e^{-\eta z_1}-2}{2\eta^2}, 
\end{array}
\]
and
\[
\begin{array}{rcl}
{\mathcal C}_{1\eta} \circ pr_1 &=& \displaystyle e^{\eta (x_1-z_1)}x_2x_3-(x_2z_3+x_3z_2)+e^{-\eta (x_1-z_1)}z_2z_3\\[6pt] 
&&+\displaystyle \frac{e^{\eta (x_1-z_1)}+e^{-\eta (x_1-z_1)}-2}{\eta^2}, \\[8pt] 
{\mathcal C}_{1\eta} \circ pr_2 &= & e^{\eta z_1}(z_2z_3)+\displaystyle \frac{e^{\eta z_1}+e^{-\eta z_1}-2}{\eta^2}, \\
\end{array}
\]
respectively.

On the other hand, the Hamiltonian functions $H_{0\eta}$ and $H_{1\eta}$ read
 \[
H_{0\eta}= e^{\eta x_1}\,x_2\,x_3+\displaystyle \frac{e^{\eta x_1}+e^{-\eta x_1}-2}{\eta^2}, \; \; H_{1\eta}=-\displaystyle e^{\eta x_1} \frac{x_2^2+x_3^2}{2}-\displaystyle \frac{e^{\eta x_1}+e^{-\eta x_1}-2}{2\eta^2}
.\\
 \]
In the new coordinates $(x, z)$, the two completely integrable Hamiltonian systems $(G_\eta \times G_\eta, \{\cdot, \cdot\}_{0\eta}, H_{0\eta})$ and $(G_\eta \times G_\eta, \{\cdot, \cdot\}_{1\eta}, H_{1\eta})$ become
\begin{equation}\label{coupledH_0-eta-Euler}
\begin{array}{rcl}
\dot{x_1}&=& e^{\eta x_1}(x_2^2-x_3^2),\\
\dot{x_2}&=& \eta  e^{\eta x_1} x_2  x_3^2-\displaystyle\frac{1}{2}\eta e^{\eta x_1}(x_2^3-x_2 x_3^2)+\displaystyle  \frac{e^{\eta x_1}-e^{-\eta x_1}}{2 \eta}(2 x_3-x_2),\\[5pt]
 \dot{x_3}&=& -\eta  e^{\eta x_1} x_2^2  x_3+\displaystyle\frac{1}{2}\eta e^{\eta x_1}(x_2^2 x_3 +x_3^3)+\displaystyle  \frac{e^{\eta x_1}-e^{-\eta x_1}}{2 \eta}(x_3-2 x_2),\\
 \dot{z_1}&=& e^{\eta z_1}(x_2z_2-x_3z_3),\\[4pt]
  \dot{z_2}&=&\displaystyle \left(\eta e^{\eta x_1}x_2x_3+\frac{e^{\eta x_1}-e^{-\eta x_1}}{\eta}\right) z_3+ x_2\frac{ e^{-\eta z_1}-e^{\eta z_1}}{2\eta}+\eta x_2e^{\eta z_1}\displaystyle (\frac{z_2^2}{2}+\frac{z_3^2}{2}),\\[8pt]
    \dot{z_3}&=&\displaystyle \left(\eta e^{\eta x_1}x_2x_3+\frac{e^{\eta x_1}-e^{-\eta x_1}}{\eta}\right) (-z_2)+ x_3\frac{ e^{\eta z_1}-e^{-\eta z_1}}{2\eta}+\eta x_3e^{\eta z_1}\displaystyle (\frac{z_2^2}{2}+\frac{z_3^2}{2}),\\

\end{array}
\end{equation}
and
\begin{equation}\label{coupledH_1-eta-Euler}
\begin{array}{rcl}
\dot{x_1}&=& e^{\eta x_1}(x_2^2-x_3^2),\\
\dot{x_2}&=& \eta  e^{\eta x_1} x_2  x_3^2-\displaystyle\frac{1}{2}\eta e^{\eta x_1}(x_2^3-x_2 x_3^2)+\displaystyle  \frac{e^{\eta x_1}-e^{-\eta x_1}}{2 \eta}(2 x_3-x_2),\\[5pt]
 \dot{x_3}&=& -\eta  e^{\eta x_1} x_2^2  x_3+\displaystyle\frac{1}{2}\eta e^{\eta x_1}(x_2^2 x_3 +x_3^3)+\displaystyle  \frac{e^{\eta x_1}-e^{-\eta x_1}}{2 \eta}(x_3-2 x_2),\\
 \dot{z_1}&=& e^{\eta z_1}(x_2z_2-x_3z_3),\\[4pt]
  \dot{z_2}&=&-\displaystyle \left(\eta e^{\eta x_1}\frac{x_2^2+x_3^2}{2}+\frac{e^{\eta x_1}-e^{-\eta x_1}}{2\eta}\right) z_2+x_3\frac{ e^{\eta z_1}-e^{-\eta z_1}}{\eta}+\eta e^{\eta z_1}x_3z_2z_3\\[8pt]  
 \dot{z_3}&=&\displaystyle \left(\eta e^{\eta x_1}\frac{x_2^2+x_3^2}{2}+\frac{e^{\eta x_1}-e^{-\eta x_1}}{2\eta}\right)z_3-x_2\frac{ e^{\eta z_1}-e^{-\eta z_1}}{\eta}-\eta e^{\eta z_1}x_2z_2z_3,\\
\end{array}
\end{equation}
respectively.

Finally, as in the Lorenz system, the multiplication $\cdot_\eta: G_\eta \times G_\eta \to G_\eta$ in the new variables is just the first projection, that is,
\begin{equation}\label{multiplication-Euler}
+((x_1, x_2, x_3), (z_1, z_2, z_3)) = (x_1, x_2, x_3).
\end{equation}
So, using (\ref{Euler.deformation}), (\ref{coupledH_0-eta-Euler}), (\ref{coupledH_1-eta-Euler}) and (\ref{multiplication-Euler}), we directly deduce that the two completely integrable Hamiltonian systems $(G_\eta \times G_\eta, \{\cdot, \cdot\}_{0\eta}, H_{0\eta})$ and $(G_\eta \times G_\eta, \{\cdot, \cdot\}_{1\eta}, H_{1\eta})$ admit reduction, via the multiplication $\cdot_{\eta} : G_\eta \times G_\eta \to G_\eta$, to the bi-Hamiltonian system on $G_\eta$ considered in Section \ref{Con-bi-Ha-sy}.


\section{Concluding remarks}

In this paper we have presented the generalization of the Poisson coalgebra construction of integrable deformations of Hamiltonian systems to the case when the initial system is bi-Hamiltonian under a pair of Lie-Poisson structures. Moreover,  the method here presented allows the systematic construction, under certain conditions, of pairs of coupled completely integrable Hamiltonian systems on Poisson-Lie groups. In this way, integrable deformations of the Lorenz and Euler top systems have been explicitly constructed.

It is worth recalling that a complete classification of the Lie-Poisson completely integrable bi-Hamiltonian systems on $\mathbb{R}^3$, which have non-transcendental integrals of motion, may be found  in~\cite{HoPe} (see also \cite{GuNu}).  In fact, the Euler top system is labeled with the number (6) in Table 1 of~\cite{HoPe}.
On the other hand, a complete classification of non-equivalent adjoint $1$-cocycles on Lie algebras of dimension $3$, whose dual maps satisfy the Jacobi identity ({\em i.e.} a complete classification of non-equivalent Lie bialgebras of dimension $3$) may be found in \cite{Go}.

By using the results in~\cite{HoPe} and~\cite{Go}, it can be easily proven by direct inspection that the method here presented could be also straightforwardly applied to the Lie-Poisson completely integrable bi-Hamiltonian systems on $\mathbb{R}^3$ which are labelled as (2), (4) and (5) in Table 1 of~\cite{HoPe}, since these are the only cases for which a common 1-cocycle does exist. Therefore, for all these cases we could obtain pairs of completely integrable systems on the product of a certain Lie-Poisson group $G_\eta$ with itself, whose projection, via the group multiplication, leads to a completely integrable bi-Hamiltonian system on $G_\eta$. The search for other Lie-Poisson bi-Hamiltonian systems on $\mathbb{R}^N$ (with $N\geq 4$) and their compatible 1-cocycles is currently under investigation.

Finally, we stress that it would be interesting to get a deeper insight into the underlying geometric structure of the (quasi) bi-Hamiltonian systems on Poisson-Lie groups that have arised in the present paper. Work on this line is also in progress and will be presented elsewhere.

\end{document}